\documentclass[12pt, reqno]{amsproc}

\usepackage[margin = 1in]{geometry}
\usepackage{amsmath,amsfonts,amssymb}
\usepackage{xcolor}
\usepackage[osf]{mathpazo}
\usepackage[euler-digits]{eulervm}
\usepackage{graphicx}
\usepackage{algorithm}
\usepackage{algpseudocode}
\usepackage{hyperref}

\hypersetup{
    colorlinks,	%
    citecolor=blue,%
    filecolor=black,%
    linkcolor=red,%
    urlcolor=gray
}

\numberwithin{equation}{section}

	\newtheorem{theorem}[equation]{Theorem}
	\newtheorem{lemma}[equation]{Lemma}
	\newtheorem{corollary}[equation]{Corollary}
	
	\newtheorem*{theorem*}{Theorem}

        \newtheorem*{fctda}{Fundamental Conundrum of Topological Data Analysis}

\theoremstyle{definition}
	
	\newtheorem{example}[equation]{Example}

\theoremstyle{remark}
	
	\newtheorem*{remark*}{Remark}

\newcommand{\R}{\mathbb{R}}

\newcommand{\N}{\mathbb{N}}

\newcommand{\Comp}{\mathcal{C}}

\newcommand{\isom}{\cong}

\newcommand{\st}{\ \mid \ }

\newcommand\abs[1]{\lvert#1\rvert}
\newcommand\norm[1]{\lVert#1\rVert}

\newcommand{\Output}{\mathcal{O}}

\DeclareMathOperator{\Dgm}{Dgm}	
\DeclareMathOperator{\im}{im}
\DeclareMathOperator{\supp}{supp}
\DeclareMathOperator{\pers}{pers}
\DeclareMathOperator{\sign}{sign}

\newcommand{\vect}{}

\begin{document}

\title{Stabilizing the unstable output of persistent homology computations}

\author{Paul Bendich}
\address{Department of Mathematics, Duke University, and Geometric Data Analytics, Inc.}
\email{bendich@math.duke.edu}
\author{Peter Bubenik}
\address{Department of Mathematics, University of Florida}
\email{peter.bubenik@ufl.edu}
\author{Alexander Wagner}
\address{Department of Mathematics, University of Florida}
\email{wagnera@ufl.edu}


\begin{abstract}
We propose a general technique for extracting a larger set of stable information from persistent homology computations than is currently done. The persistent homology algorithm is usually viewed as a procedure which starts with a filtered complex and ends with a persistence diagram. This procedure is stable (at least to certain types of perturbations of the input). This justifies the use of the diagram as a signature of the input, and the use of features derived from it in statistics and machine learning. However, these computations also produce other information of great interest to practitioners that is unfortunately unstable. For example, each point in the diagram corresponds to a simplex whose addition in the filtration results in the birth of the corresponding persistent homology class, but this correspondence is unstable. In addition, the persistence diagram is not stable with respect to other procedures that are employed in practice, such as thresholding a point cloud by density. We recast these problems as real-valued functions which are discontinuous but measurable, and then observe that convolving such a function with a suitable function produces a Lipschitz function. The resulting stable function can be estimated by perturbing the input and averaging the output. We illustrate this approach with a number of examples, including a stable localization of a persistent homology generator from brain imaging data.
\end{abstract}

\maketitle

\section{Introduction} \label{sec:intro}

Persistence diagrams, also called bar codes, are one of the main tools in topological data analysis \cite{Carlsson2009,FrosLand99,Edelsbrunner2010,ghrist:survey}. In combination with machine-learning and statistical techniques, they have been used in a wide variety of real-world applications, including the assessment of road network reconstruction \cite{Ahmed2014Roads}, neuroscience \cite{cbk:ipmi2009}, \cite{Bendich2015trees}, vehicle tracking \cite{Bendich2015tracking}, object recognition \cite{chunyuan:2014}, protein compressibility \cite{Gameiro:2015b}, and protein structure~\cite{giseon:maltose}.

Put briefly, these persistence diagrams are multi-sets of points in the extended plane, and they compactly describe some of the multi-scale topological and geometric information present in a high-dimensional point cloud,
or carried by a real-valued function on a domain.
Several theorems \cite{CohenSteiner2007,Chazal2009b,CohenSteiner2010} state that persistence diagrams are stable with respect to certain variations in the point-cloud or functional input, and so the conclusions drawn from them can be taken with some confidence.

On the other hand, there is additional potentially very useful but unstable information produced during the computation of persistence diagrams. 
For example, a point far from the diagonal in the degree-zero persistence diagram represents a connected component with high persistence. This component first appears somewhere and the computation that produces the persistence diagram can be used to find its location. However this location is not stable: as we will describe below, a small change in the input will cause only a small change in the persistence of this connected component, but it can radically alter the location of its birth.
We summarize this as follows.

\begin{fctda} 
  Users of topological data analysis would like to find the simplices or cycles corresponding to the birth of the most significant pairings of critical values. However, unlike the paired critical values, these simplices and cycles are unstable.
\end{fctda}

In addition, persistent homology computations may rely on parameters such that the output persistence diagram is not stable with respect to changes of these parameters.

\subsection{Our Contribution}
\label{sec:contribution}

This paper introduces a method for stabilizing desirable but unstable outputs of persistent homology computations. The main idea is the following. On the front end, we think of a persistent homology computation $\Comp$ as being parametrized by a vector $\vect{a} = (a_1, \ldots, a_n)$ of real numbers. These parameters could specify the input to the computation (e.g. the coordinates of the vertices of a simplicial complex) or they could specify other values used in the computation (e.g. threshold parameters used in de-noising or bandwidths for smoothing). For a given choice of $\vect{a}$, we get a persistence diagram. On the back end, we consider a function $p$ that extracts a real-number summary from a persistence diagram. For example, $p$ might extract the persistence of a homology class created by the addition of a specific edge in a filtered simplicial complex, or it might be an indicator function on whether or not the longest bar was born by the addition of a simplex contained in a fixed region of the input space, or it may indicate whether or not a chosen representative geometric cycle intersects a given region. 
The composite function $h$ that maps the parameter vector to the real number  need not be continuous, but it will in many cases be {measurable}. 
We convolve this function with a Gaussian (or indeed any Lipschitz function) to produce a new Lipschitz function that carries the persistence-based information we desire.

Our main theoretical results (Theorems \ref{thm:stability1}, \ref{thm:stability2}, and \ref{thm:stability3}) give conditions on functions $h$ and $K$ (where $K$ will usually be a kernel) that guarantee that the convolution $h * K$ is Lipschitz with specified Lipschitz constant.
From these we obtain the following, where more precise statements are given as Corollaries \ref{cor:triangular}, \ref{cor:epanechnikov}, and \ref{cor:gaussian}.

\begin{theorem}
  If $h$ is locally essentially bounded then for the triangular and Epanechnikov kernels, $h * K$ is locally Lipschitz.
  If $h$ is essentially bounded then for the Gaussian kernel, $h * K$ is Lipschitz.
\end{theorem}

In practice, this can be translated to
{a simple procedure for stabilizing unstable persistent homology computations: perturb the input by adding, for example, Gaussian noise, and redo the computation; repeat and average}. See Algorithm~\ref{alg:main}. By the law of large numbers, the result converges to the desired stable value.

\begin{theorem}
  Let $\vect{\epsilon}_1,\ldots,\vect{\epsilon}_M$ be drawn independently from a kernel $K$. Then
  \begin{equation*} 
    \frac{1}{M} \sum_{i=1}^M h(\vect{a} - \vect{\epsilon}_i) \to (h*K)(\vect{a}).
  \end{equation*}
\end{theorem}

\begin{algorithm}
  \caption{Stabilizing unstable persistence computations}
  \label{alg:main}
  \begin{algorithmic}
    \Require $h:\R^n \to \R$, $a \in \R^n$
    \Ensure $M \in \N$, $\sigma > 0$
    \For {$i \leftarrow 1, M$}
      \For {$j \leftarrow 1, n$}
        \State Sample $\epsilon_j$ from $N(0,\sigma^2)$
      \EndFor
      \State $y_i \leftarrow h(a + \epsilon)$, $\epsilon = (\epsilon_1,\ldots,\epsilon_n)$
    \EndFor
    \State \textbf{return} the average value of $y_1,\ldots,y_M$  
  \end{algorithmic}
\end{algorithm}

We summarize our computational pipeline in the following algorithm.
Say we have performed a persistence computation and obtained an unstable output. For example, we have determined that the longest interval in the degree one bar code of the Vietoris-Rips complex on points $X_1,\ldots,X_N \in \R^d$ is born with the addition of the edge $X_1X_2$.
We encode this output as a function $h:\R^n \to \R$ with input $a \in \R^n$.
For example, the coordinates of the above points give us $a \in \R^n$ where $n=Nd$. We define a function $h:\R^n \to \R$ whose value is the length of the longest interval in the bar code if it is born with the addition of the edge $X_1X_2$ and is otherwise $0$. 
 
The choice of standard deviation $\sigma$ (also called bandwidth) is discussed in Sections \ref{sec:bandwidth-theory} and \ref{sec:bandwidth-practice}. In Section~\ref{sec:stability-kernel}, we prove that Algorithm~\ref{alg:main} is stable with respect to this choice.

\subsection{Three examples}

For the reader familiar with persistent homology who wants to see how this works in practice, we provide three examples, 
the first and third to synthetic data and 
the second to brain imaging data.
Code for these examples is available at 
\url{https://github.com/peter-bubenik/stabilizing-paper-code}.

\begin{example}
\label{ex:double-annulus}
\emph{A cycle generating a persistent homology class}

\begin{figure}
  \centering
  \includegraphics[width=\textwidth]{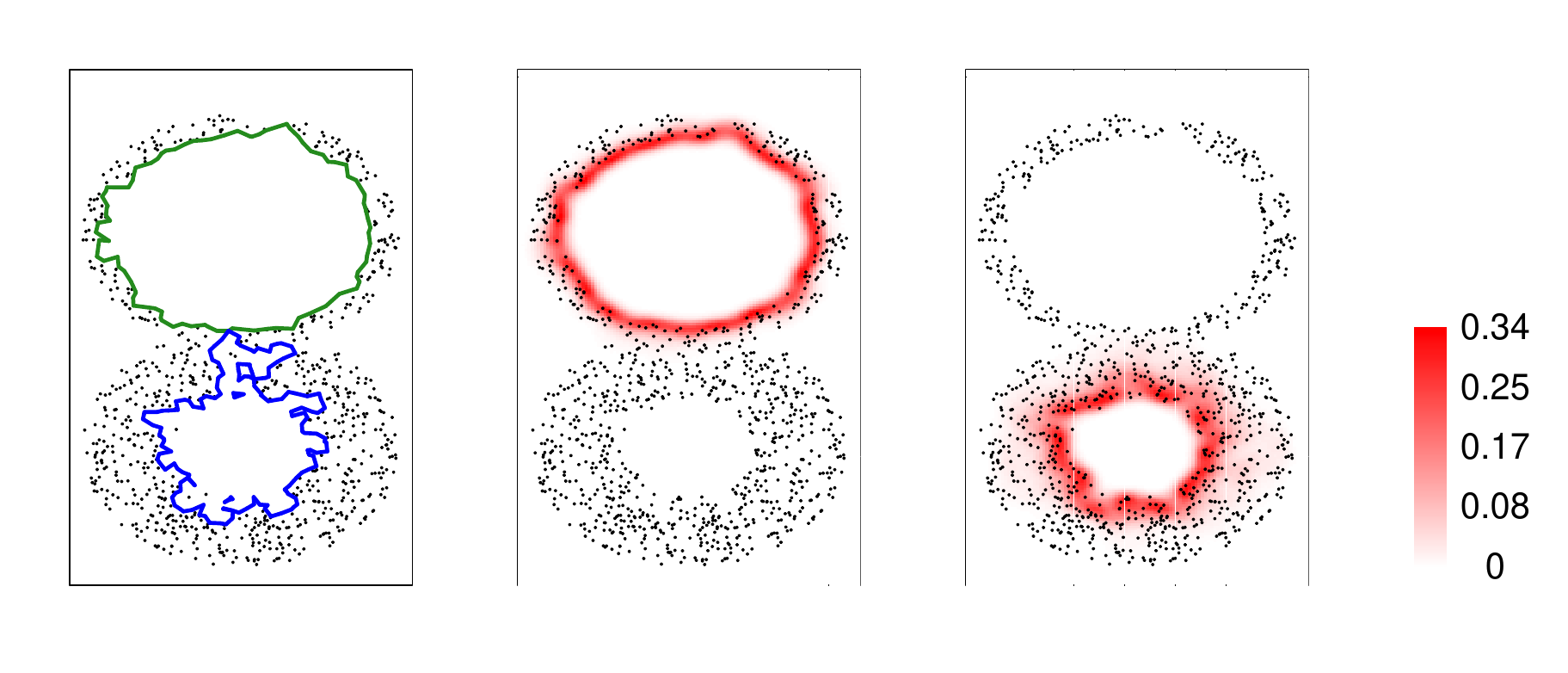}
  \caption{Finding the first and second longest bars. The first panel shows the original sample and the representative cycles produced by Dionysus~\cite{dionysus} for the first and second longest bars. The latter two panels show the proportion of perturbations for which each square in a grid intersected the representative geometric cycle of the first or second longest bar, respectively.}
  \label{fig:doubleannulus}
\end{figure}

We sample 1000 points uniformly from two conjoined annuli of inner and outer radii $(20, 50)$ and $(40, 50)$. 
Using Dionysus~\cite{dionysus}, we compute the $1$-dimensional persistent homology of the alpha complex of our sample and obtain a representative cycle for the longest and second-longest bars.
See Figure~\ref{fig:doubleannulus}, left panel.
However, the embedded location of these cycles is unstable. 
We would like to quantify the uncertainty of this location. 
To do so, we consider a square grid with edge-length $1$.
Our function $h: \R^{2000} \to \R$ has input the coordinates of the sampled points and has output $1$ if the geometric cycle produced by Dionysus intersects a given square in our grid and otherwise has output $0$.

We perturb the sampled points 10,000 times by adding Gaussian noise with standard deviation $3$.
For each square, we find the proportion of trials in which the representative geometric cycle for the longest or second-longest bar produced by Dionysus intersects the square. 
By performing this procedure simultaneously for every square in the grid, we obtain the second and third panels in Figure~\ref{fig:doubleannulus}. 

To see the effect of varying the choice of bandwidth, see Section~\ref{sec:bandwidth-practice}.
\end{example}

\begin{example}
\label{ex:brain-artery}
\emph{Location of a persistent homology generator in brain imaging data}

In~\cite{Bendich2015trees}, the authors apply topological data analysis to brain arteries extracted from magnetic resonance images.
Mathematically, each of these brain arteries is a graph embedded in three-dimensional Euclidean space.
Using the height (the $z$-coordinate) one obtains a filtration on this graph, which may be used to compute degree-zero persistent homology.

To facilitate statistical analysis of the resulting persistence diagrams (i.e. bar codes), they convert each persistence diagram to a vector consisting of the lengths of the 100 longest bars in decreasing order.
In their analysis, the length of the $28$th longest bar is a numerical feature that yields a correlation with age that is near-optimal among vector features consisting of the lengths of the $i$th through $j$th longest bars for any $1 \leq i \leq j \leq 100$.

If one wants to find a biological interpretation of this result, it is obvious to ask for the location of the generator of the $28$th longest bar for each subject.
 It is easy to locate the generator responsible for the birth of the $28$th longest bar. It will be a particular vertex of the graph, whose image is a point in space.
However, the location of this point is unstable: as we will later explain, small perturbations of the spatial coordinates of the vertices of the graph can lead to large changes of this location.

\begin{figure}
  \centering
  \includegraphics[width=0.35\textwidth]{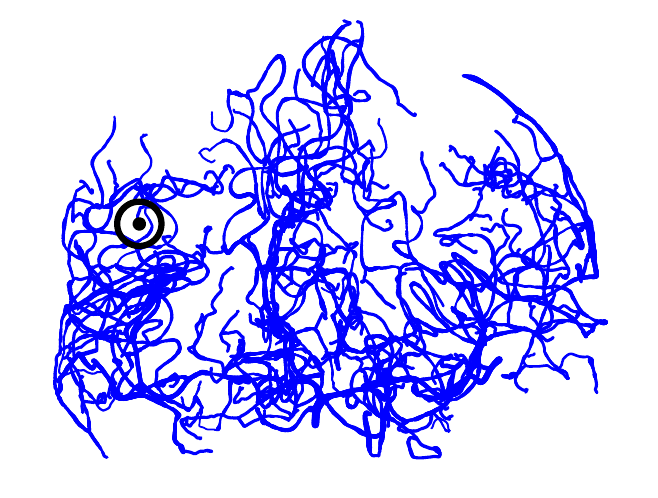}
  \caption{The brain arteries of the first subject in~\cite{Bendich2015trees}. The black dot is the location of the generator of the $28$th longest bar in degree-zero persistent homology. We consider the indicator function on the location of this generator with respect to the given sphere.}
  \label{fig:brain_artery}
\end{figure}

We choose a ball centered at this point and consider the function whose value is $1$ if the location of the generator of the $28$th longest bar is located in this ball, and is otherwise $0$.
The resulting function $h:\R^{3V} \to \R$ (where $V$ is the number of vertices in the graph) is unstable, but it may be stabilized using the method summarized in Section~\ref{sec:contribution}.
Applying Algorithm~\ref{alg:main} with $M=1000$ and $\sigma = 0.1$ we obtain an estimate of the stable value of $h*K$ evaluated at the observed input, equal to $0.637$. This shows that under small perturbations of the input, over half of the time the generator of the $28$th longest bar is located in the chosen ball.
This result holds for a large range of sizes of balls - see Section~\ref{sec:brain-location} for some further discussion.

We remark that this approach provides a resolution of the conflict between TDA theorists and TDA users expressed in the Fundamental Conundrum of Topological Data Analysis in the introduction. We can provide TDA users with a location of a generator of a persistent homology class together with an estimate of a stable real value of how often this location lies in a given region under certain perturbations.
\end{example}

\begin{example}
\label{ex:torus}
\emph{Persistence of a homology class born in a region}

Consider the function $f$ on the square in Figure~\ref{truefunction}.
This induces a function $\bar{f}$ on the torus since $f(x,y)=0$ on the boundary of the square. Suppose we are only given a finite sample of this induced function and we are interested in the presence of long-lived bars which are born in the region of the torus corresponding to the second quadrant of the square.

\begin{figure}
  \centering\includegraphics[width=45mm]{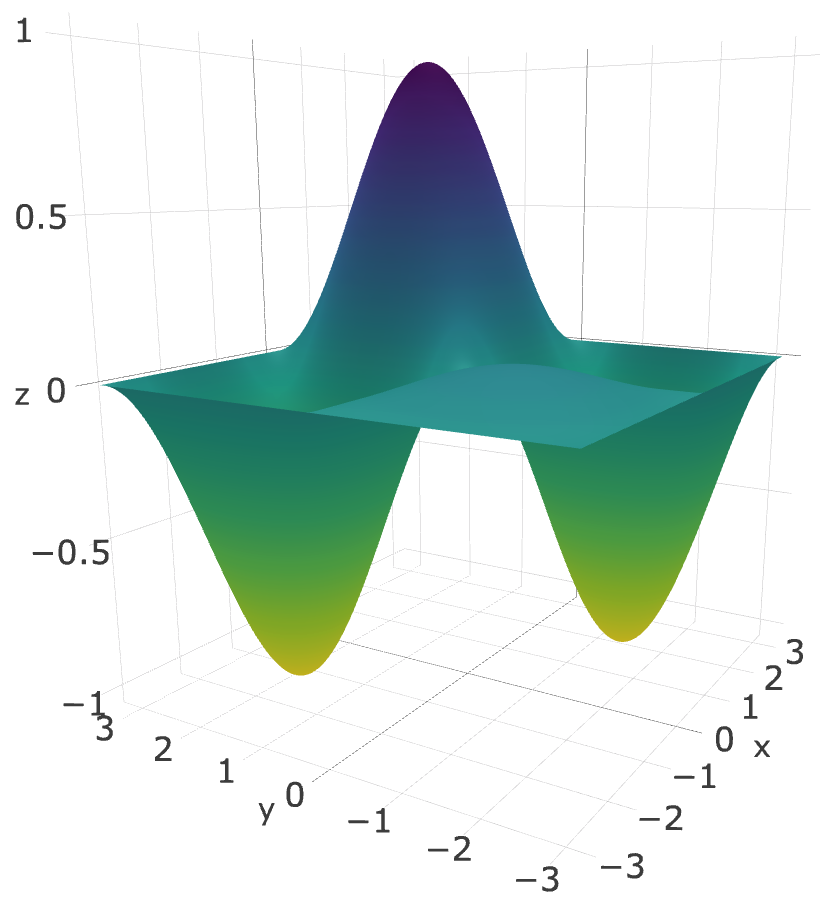}
  \caption{The graph of the function on the square $[-\pi,\pi]^2$
    given by
    $f(u,v) = \sin(u)\sin(v)(1-0.9^*1_{\{u<0,v<0\}}):[-\pi,\pi]^2\to
    \mathbb{R}$. It induces a function on the torus, $\bar{f}:T^2 \to \R$, with two global minima with value $-1$, one global maximum with value $1$, one local maximum with value $0.1$, and four saddle points with value $0$. From}
  \label{truefunction}
\end{figure}

To be concrete, we start with a sample $X$ of $N$ points from the graph of $\bar{f}$, by sampling $u_i,v_i$ independently from the uniform distribution on $[-\pi,\pi]$ and letting $z_i = f(u_i,v_i)$.
Note that $X$ is a random variable.
We use $X$ to construct a filtered simplicial complex approximating the unknown function $\bar{f}$ as follows. 
From the points $\{(u_i,v_i)\}$ we construct a Delaunay triangulation of the torus. We filter this simplicial complex by assigning the vertex $(u_i,v_i)$ the value $z_i$ and assigning edges and triangles the maximum value of their vertices.

\begin{figure}
  \includegraphics[width=0.33\textwidth]{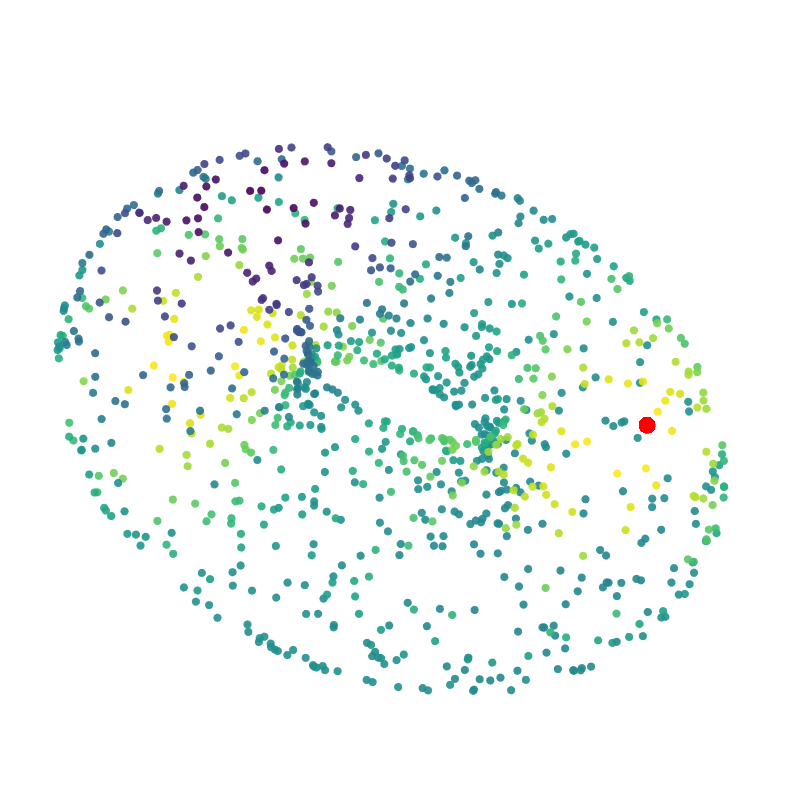}
  \includegraphics[width=0.33\textwidth]{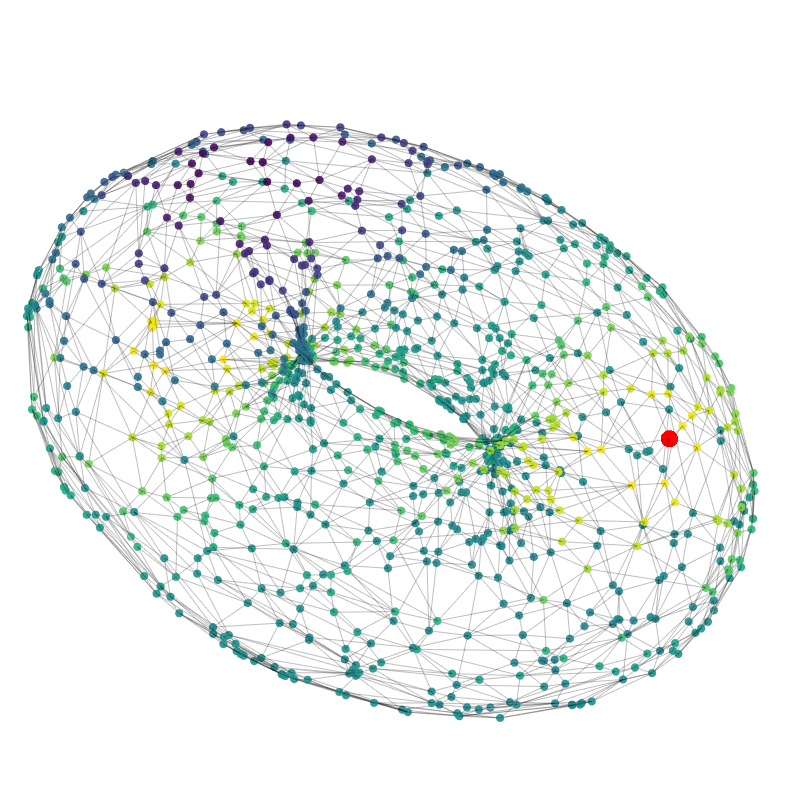}
  \includegraphics[width=0.33\textwidth]{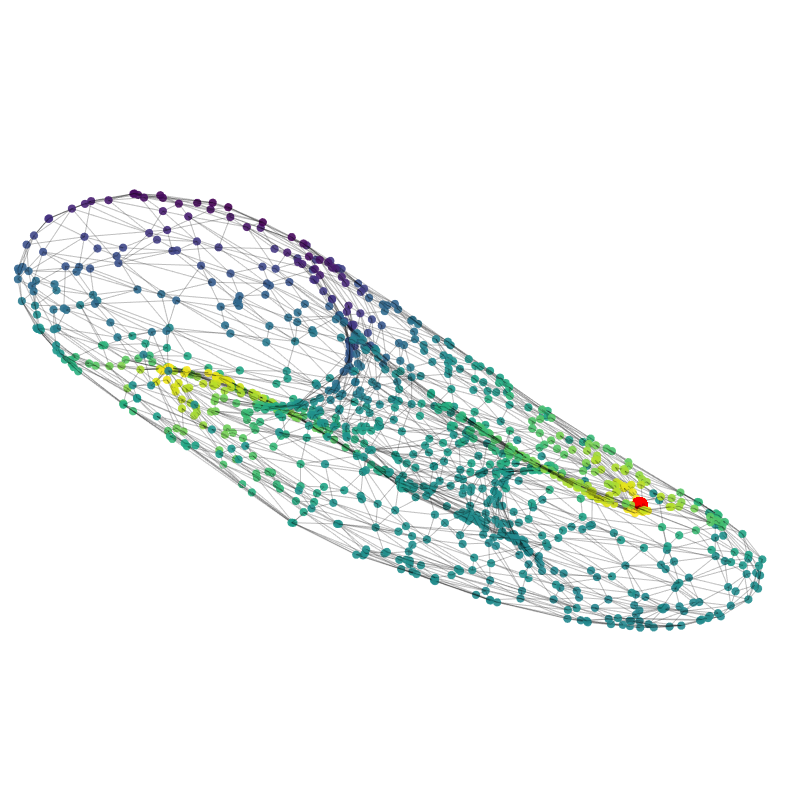}
  \caption{Sample of 1000 points from the graph
    $\{(x,f(x)) : x\in T^2\}$, where the function values are indicated
    using the same color scale as in Figure~\ref{truefunction}. The
    points on the torus are used to construct a Delaunay
    triangulation, which is filtered using the function values. On the
    right we indicate the filtration values by moving the points in
    the normal direction.}
  \label{fig:sample}
\end{figure}

We compute the $0$-dimensional extended\footnote{Extended persistent homology follows the homology of increasing sublevel sets with the relative homology of the whole space relative to decreasing superlevel sets~\cite{cseh:extendingP}. In the case considered here, it pairs the global minimum with the global maximum.} persistence diagram of this filtered simplicial complex.
Let $h(X)$ be the length of the longest bar if that bar was born in the region corresponding to the second quadrant (see Figure~\ref{truefunction}) and $0$ otherwise.

This process defines a function $h:\mathbb{R}^{3N}\to\mathbb{R}$, but $h$ is unstable. Consider the sample $X=x$ in Figure~\ref{fig:sample}.
We have $h(x)=0$ since the global minimum, highlighted in red, is born outside the region corresponding to the second quadrant. Because of the symmetry of $f$, the random variable $h(X)$ is $0$ approximately half the time and about $2$ approximately half the time. 

Let $K$ denote the $3N$-variate Gaussian with mean $0$ and standard deviation $0.2$.
For $M \geq 1$, sample $\epsilon_1,\ldots,\epsilon_M$ independently from $K$.
Compute $\frac{1}{M} \sum_{i=1}^M h(x-\epsilon_i)$. 
See Figure~\ref{fig:long-bars}.
As $M$ increases, this quantity converges to $g(x)$, where $g := h * K$ is the stabilized version of $h$.

\begin{figure}
  \centering\includegraphics[width=60mm]{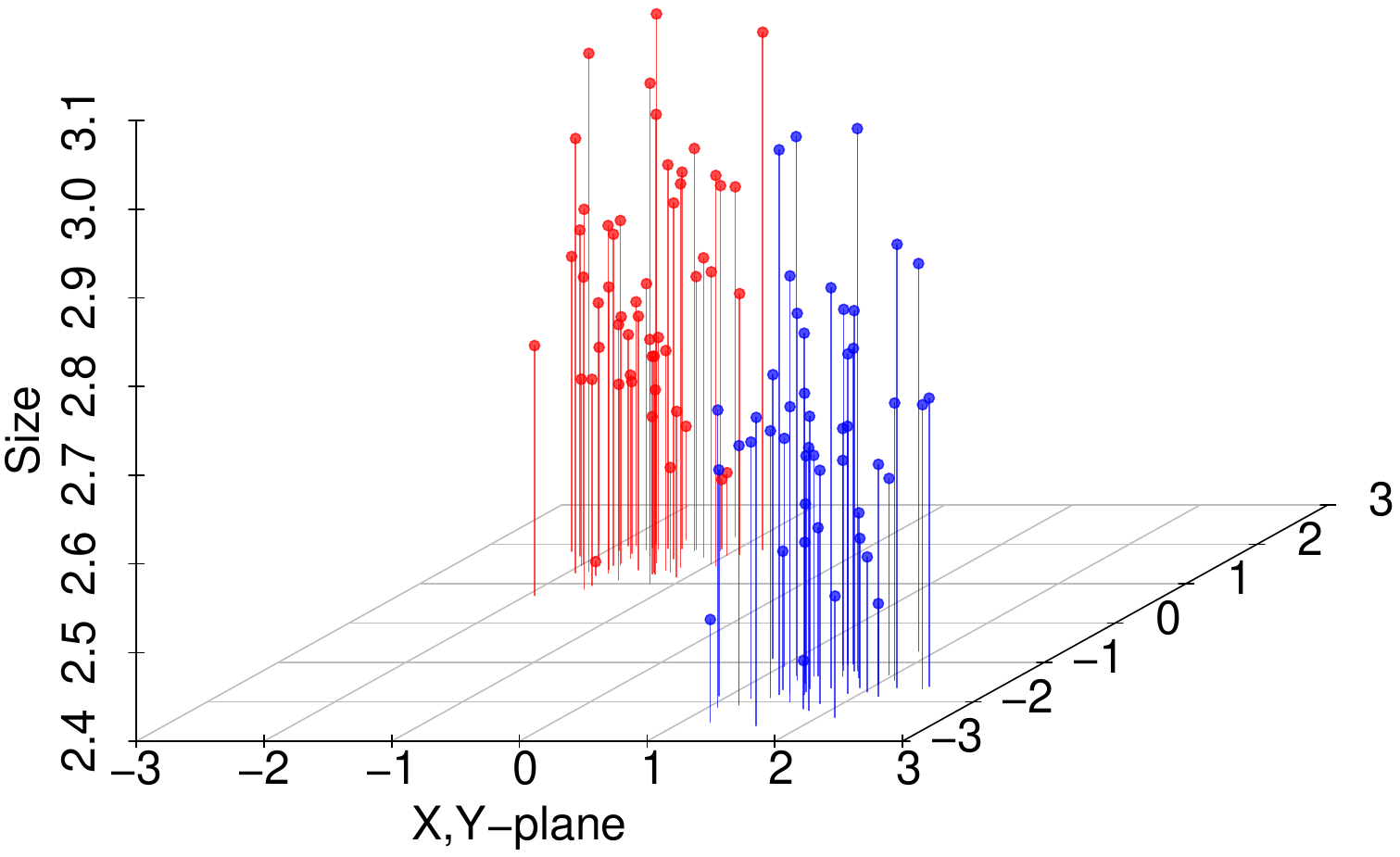}
  \caption{Locations and sizes of 100 longest bars from the trials. Averaging the lengths of the red bars over 1000 trials we get $1.291$, which is consistent with the fact that the random variable $h(X)$ is $0$ or about $2$ with equal probability. 
We should not expect $\lim_{M\to\infty}\frac{1}{M}\sum_{i=1}^{M}h(x+\epsilon_i)$ to converge to $1$ because unlike $f$, a particular sample $X=x$ is not symmetric with respect to the second and fourth quadrants.}
  \label{fig:long-bars}
\end{figure} 

\end{example}

\subsection{Related Work}

Partial inspiration for the main idea of our work (when faced with an instability caused by a near-interchange of values, perturb the values many times and take some sort of average) comes from the trembling-hand equilibrium solution \cite{munch2015probabilistic} to the non-uniqueness problem for Fr\'echet means of persistence diagrams. 
Our approach should also be compared with the topological reconstruction results of Niyogi, Smale and Weinberger~\cite{nsw2}.

Several recent papers have advocated principled approaches for extracting features from persistence diagrams, including persistence landscapes \cite{Bubenik2015landscapes}, the stable multi-scale kernel~\cite{Reininghaus2015kernel}, intensity functionals \cite{Chen2015intensity}, persistence images~\cite{Adams:2017}, the stable topological signature~\cite{Carriere2015signatures}, and the cover-tree entropy reduction~\cite{Smith:2017}. Our result complements these ideas: once one identifies some specific parts of the persistence diagram as having good classification power, one can then attempt to locate, in a robust way, the portions of the domain responsible for these parts.
Other papers (e.g. \cite{Chazal2011measure,bckl:nonparametric,adams:2011}) have developed sophisticated schemes for data-cleaning before persistent homology computation.
These techniques are generally fragile to certain initial parameter choices, such as the $m_0$ parameter in \cite{Chazal2011measure}. Again, we provide a complementary role: any of these schemes can be run many times for several perturbations of an initial parameter choice, and the output can then be taken with confidence.

Dey and Wenger~\cite{DeyWenger:2007} have shown that the critical points of interval persistent homology are stable in the sense that they remain within some path-connected component.

Zomorodian and Carlsson~\cite{Zomorodian2008localizing} use Mayer-Vietoris as inspiration in their technique for localizing (relative to a cover) homology classes within a given simplicial complex. However, this works only for a fixed simplicial complex, not a simplicial complex endowed with a filtration, and the results are certainly fragile to changes in this fixed complex.

Weinberger~\cite{Weinberger:2014} considers the sample complexity of some basic problems of topological inference. Specifically, he estimates the number of sample points necessary to determine the dimension, topological type, and to detect singularities for certain spaces.

Robust summaries of persistent homology are considered in the following papers; they do not consider the location of homology generators.
In~\cite{Blumberg:2012}, Blumberg et al. show that persistent homology on a metric measure space induces a stable empirical measure in the space of persistence diagrams. Taking the distance to a reference distribution or a reference barcode, they obtain robust statistics.
In~\cite{Chazal:2014d}, the authors derive limiting distributions and confidence sets for persistence diagrams based on the sub-level sets of the distance-to-a-measure. 

Convolving with a kernel to obtain smoothness is a classical idea in statistics~\cite{silverman:book,wandJones:book}. It has been used to construct smooth estimators of discrete data as an initial step to computing persistent homology~\cite{bckl:nonparametric,Bubenik2015landscapes,Fasy:2014}.
A related idea is the to perform subsampling (e.g. the bootstrap) to obtain convergence results and confidence intervals for persistence diagrams and persistence landscapes~\cite{Fasy:2014,cflrw:silhouettes,Chazal:2014c,Chazal:2014a}.
These papers use ideas related to ones presented here, but to smooth initial data or to smooth stable outputs of persistence computations, not to stabilize unstable outputs of persistence computations.

\subsection{Computational issues}

Our methods (applying Theorem~\ref{thm:lln} in Algorithm~\ref{alg:main})
require repeated computation of persistence diagrams for similar filtrations.
The computational cost may be considerable.
In Examples \ref{ex:brain-artery} and \ref{ex:torus} we repeat $M=1000$ times.
In Example~\ref{ex:double-annulus} we repeat $10,000$ times.
Note that we do not have convergence results at this time.
For repeated persistent homology calculations it is important to have efficient software. 
In Example~\ref{ex:double-annulus} we use Dionysus~\cite{dionysus},
in Examples~\ref{ex:brain-artery} and~\ref{ex:component} we calculate persistent homology using a union-find data structure~\cite{edelsbrunnerHarer:book}, in Example~\ref{ex:torus} we use Perseus~\cite{perseus}, and in Example~\ref{ex:edge} we use Ripser~\cite{ripser}.
  
However, our methods are trivially parallelizable. With access to many cores, our repeated computations can be computed in parallel without increasing the running time.

Note that for small perturbations, much of the persistent homology computation may be the same. In this case, there may be considerable computational savings by 
using vineyard updates~\cite{csem:vineyards}.

Let us also remark that our methods combine nicely with subsampling, which is crucial for allowing persistent homology computations in the big data setting~\cite{Chazal:2014c}.

\subsection{Replacing persistence diagrams with features}

Our approach centers on converting a persistence diagram to a real number. This may seem simplistic and somewhat ad hoc. However, all effective methods of combining persistence diagrams with serious statistical analysis and machine learning techniques rely on replacing a persistence diagram with a vector in some Hilbert space or Banach space. For simplicity, we restrict ourselves to the vector space $\R$, but our approach can be extended to more general vector spaces.

\subsection{Outline}

Persistent homology computations and stability theorems are reviewed in Section \ref{sec:PD}, although we assume the reader is already somewhat familiar with them. 
Several examples of important but unstable persistence-based information are given in Section \ref{sec:IG}, and we then describe a general approach that stabilizes them in Section \ref{sec:main}. 
In Section~\ref{sec:ph-computations}, we show how to apply these results to various persistent homology computations.
Additional analyses and discussion are presented in Section \ref{sec:examples}.
Potential future directions are discussed in Section \ref{sec:Disc}.

\section{Persistent Homology and Stability} \label{sec:PD}

The treatment of persistence diagrams here is adapted from \cite{Edelsbrunner2010}. For a more general discussion, see \cite{oudot:book}. We assume the reader is familiar with the basics of homology groups: the textbook \cite{Munkres2} is a good introduction. All homology groups are assumed to be computed over some fixed field.
For concreteness, we restrict our attention to simplicial complexes, but our results also apply to more general complexes.

\subsection{Persistent Homology}


Persistent homology is computed for a finite \emph{filtered abstract simplicial complex}. That is, we have a finite \emph{abstract simplicial complex}, a collection,  $K=\{\sigma\}$, of nonempty subsets of a fixed finite set that satisfy the condition that if $\emptyset \neq \tau \subseteq \sigma \in K$ then $\tau \in K$. In addition, we have a \emph{filtration}, a function $f:K \to \R$ such that $\tau \subseteq \sigma$ then $f(\tau) \leq f(\sigma)$. That is, $f$ is order preserving.


Fix a homological dimension $p$.
Suppose the distinct values of $f$ are $r_1 < \ldots < r_m.$ 
For each $1 \leq i \leq m,$ define $K^i = \{\sigma \in K \mid f(\sigma) \leq r_i\}$.
Since $f$ is order preserving, each $K^i$ is a subcomplex. Whenever $i \leq j$, there is an inclusion $K^i \hookrightarrow K^j$, which
induces a homomorphism:
\[
f_p^{i,j}: H_p(K^i) \to H_p(K^j).
\]
A homology class $\alpha \in H_p(K^i)$ is a \emph{persistent homology class} that is \emph{born} at level $i$ if $\alpha \notin \im f_p^{i-1,i}$, and that \emph{dies} entering level $j$
if $f_p^{i,j}(\alpha) =0$ but $f_p^{i,j-1}(\alpha) \neq 0$.
If $\alpha$ never dies, we say that it dies entering level $j=\infty$ and $r_{\infty} = \infty$.
The \emph{persistence} of $\alpha$ is defined to be $\pers(\alpha) = r_j - r_i.$ The set of classes which are born at $i$ and die entering level $j$ form a vector space, with rank denoted $\mu_p^{i,j}.$
The degree-$p$ \emph{persistence diagram} of $f$, $\Dgm_p(f),$ encodes these ranks. It is a multiset of points in the extended plane, with a point
of multiplicity $\mu_p^{i,j}$ at each point $(r_i,r_j).$ 


In practice, one constructs a filtered abstract simplicial complex from some other starting data. In addition, more information can be extracted from the persistent homology algorithm than just the persistence diagram.

We define a \emph{persistent homology computation}, $\Comp$, to be a function whose input consists of real numbers $a_1,a_2,\ldots,a_n$. These may include input values and also parameter values for the computation.
Using this input, $\Comp$ constructs an abstract simplicial complex $K$ together with a filtration $f$.
The output of $\Comp$ consists of a degree-$p$ persistence diagram together with for each $(r_i,r_j)$ in the persistence diagram (counted with multiplicity), 
a $p$-simplex $\sigma$ with $f(\sigma)=r_i$, 
a $p$-cycle $\alpha$ in $K^i$ containing $\sigma$,
a $(p+1)$-simplex $\tau$ with $f(\tau) = r_j$,
and a $(p+1)$-chain $\beta$ in $K^j$ containing $\tau$ with $d\beta = \alpha$.


\begin{example}
\label{ex:HF}
\emph{Functions on simplicial complexes.}
A filtered abstract simplicial complex, $K$, may be obtained from a real-valued function, $F$, on the vertices in a finite simplicial complex, $\mathcal{K}$. As a set $K \isom \mathcal{K}$. A filtration, $f$, on $K$ is defined by $f(\sigma) = \sup_{x \in \sigma} F(x)$.

\begin{figure}
  \centering
  \includegraphics[scale=0.2]{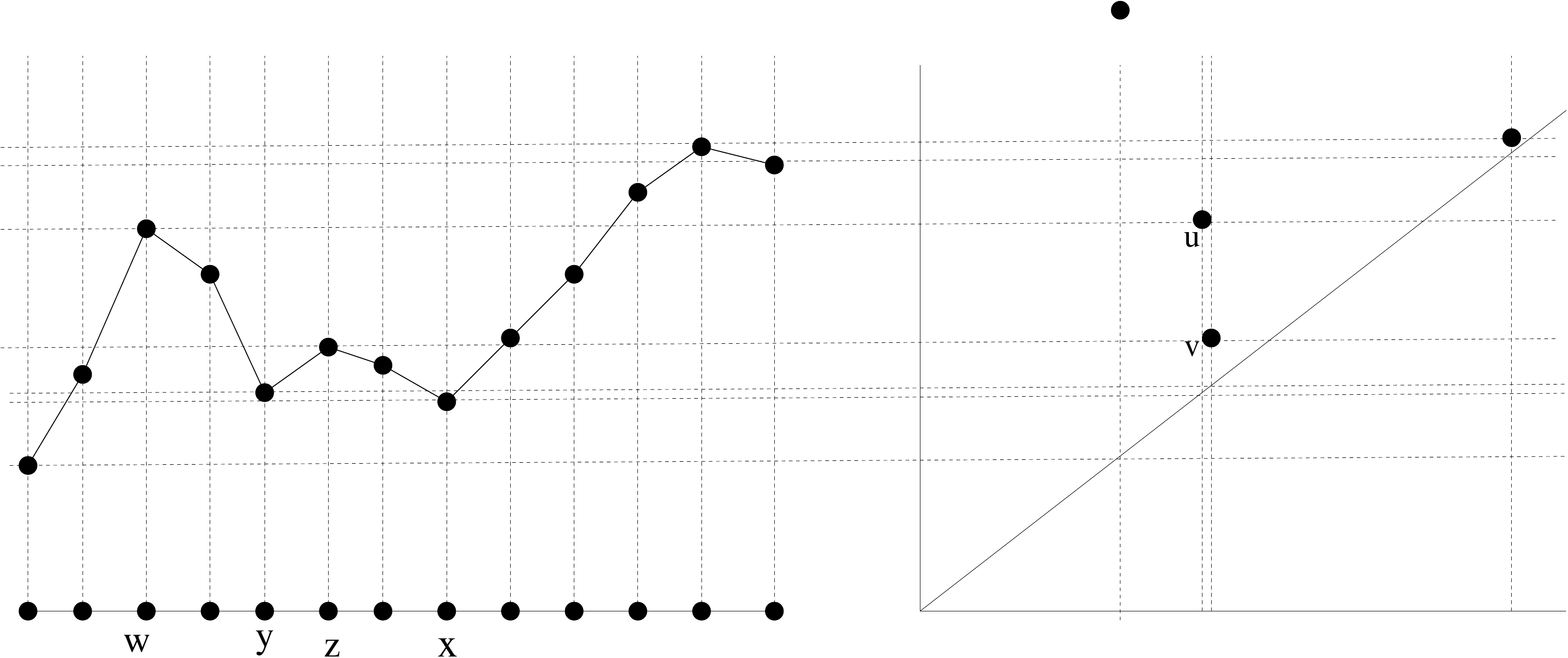}
   \caption{Left: The graph of a function $F$ on a simplicial complex $\mathcal{K}$. Right: the degree-zero persistence diagram $\Dgm_0(f)$ for the corresponding abstract simplicial complex $K$ and filtration $f$. The labeled points
have coordinates $u = (f(x),f(w))$ and $v = (f(y),f(z)).$ The point on the very top has infinite $y$-coordinate.}
\label{fig:ZeroDiag}
\end{figure}

For example, let $\mathcal{K}$ be the geometric line graph (i.e. an embedding of a graph - consisting of vertices and edges - in the plane), shown on the bottom of the left side of Figure \ref{fig:ZeroDiag}.
Above this, we have the graph of a function $F$ on the points in $\mathcal{K}$.
From this, we have a corresponding abstract simplicial complex $K$ and filtration $f$.
The persistence diagram $\Dgm_0(f)$ is on the right. 
The input to $\Comp$ consists of the function values (from left to right) $a_1,a_2,\ldots,a_n$.
\end{example}

\begin{example}
 \emph{Distance to a PL-Curve.}
\label{ex:FDC}
Consider the piecewise-linear curve $C$ on the left side of Figure \ref{fig:EC}.
Moving clockwise, we order its vertices $A = v_1, v_2, \ldots v_N = D$.
Let $K$ be the full simplex on these $N$ vertices. For each vertex $v$, define $f(v) = 0$.
For each edge of the form $e = (v_i, v_{i+1}),$ define $f(e) = 0$, and for any other edge $e = (v_i, v_j),$
we set $f(e)$ to be the Euclidean distance between $v_i$ and $v_j$. Finally, for any higher simplex $\sigma$, set $f(\sigma) = max_{e \subseteq \sigma} f(e)$, where we take the maximum over the set of edges contained in $\sigma$. The degree-one persistence diagram $\Dgm_1(f)$ appears on the right
of Figure \ref{fig:EC}. 
\begin{figure}
  \centering
  \includegraphics[scale=0.3]{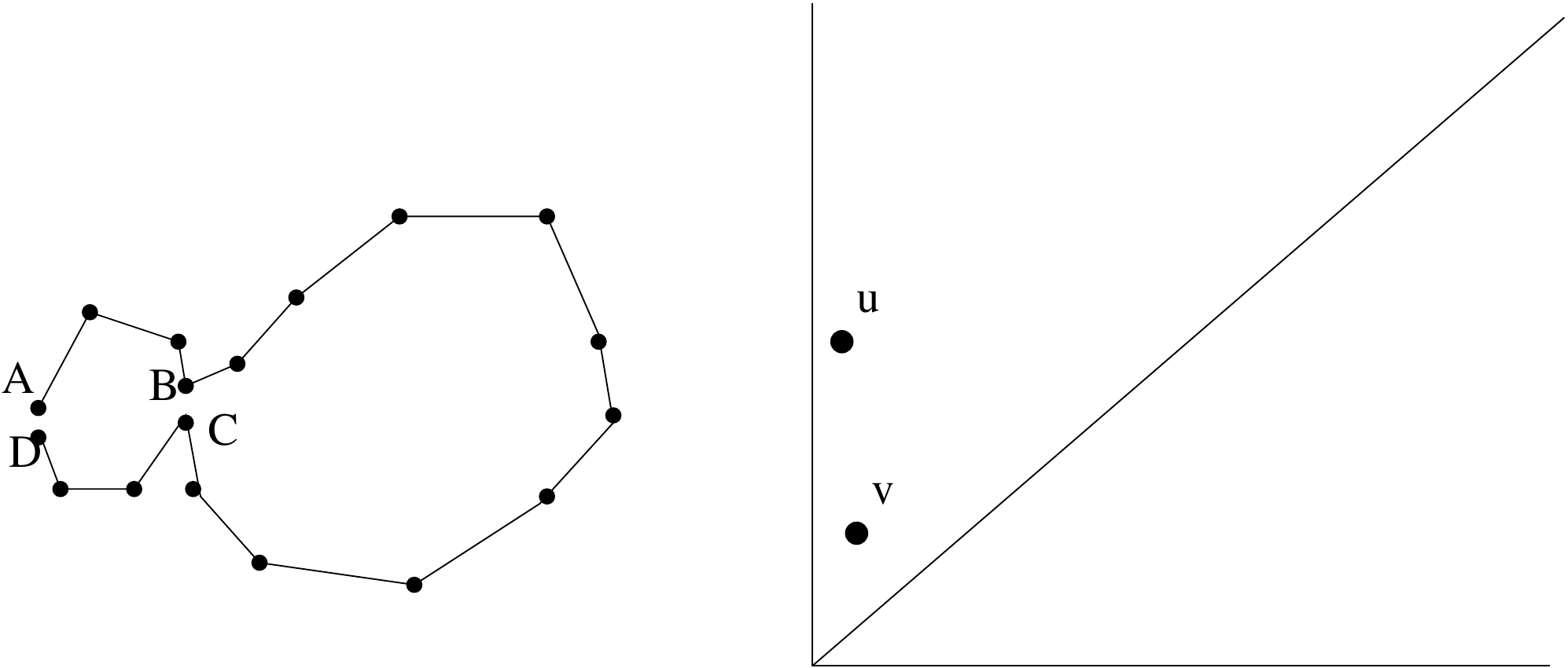}
   \caption{Left: a piecewise-linear curve in the plane. The distance between $A$ and $D$ is slightly smaller than
the distance between $B$ and $C$. Right: $\Dgm_1(f),$ where $f$ is as defined in the text. The points $u$ and $v$ correspond to one-cycles that are
created by the additions of edges $(A,D)$ and $(B,C)$, respectively.}
\label{fig:EC}
\end{figure}
\end{example}
Here the input to $\Comp$ consists of the $2n$ coordinates of the vertices.
We note this paradigm can be extended to curves $C$ in higher-dimensional ambient spaces, or even to higher-dimensional complexes.

%

\begin{example}
\label{ex:RF}
 \emph{Point cloud -- Vietoris-Rips.}
Suppose that $X = \{x_1,\ldots,x_N\}$ is a set of points in some metric space $(Y,d)$. 
We let $K$ be the full simplex on these vertices. Define $f(v) = 0$ for each vertex and $f(e) = d(v,w)$ for each edge $e = (v,w).$
As above, we set $f(\sigma) = max_{e \subseteq \sigma} f(e)$ for all higher-dimensional simplices.
This is called the Vietoris-Rips filtration.
We denote $\Dgm_p(X) = \Dgm_p(f)$. 
The input to $\Comp$ consists of the coordinates of the points in $X$ in some parametrization of $Y$. 
Alternatively, it consists of the entries of the distance matrix $D = (d(x_i,x_j))$.

For example, let $X$ be the annular point cloud on the top-left of Figure \ref{fig:FuzzyOutliers}. The corresponding $\Dgm_1(X)$ appears on the top-right of the same figure.
\end{example}

\begin{example}
  \label{ex:geometry}
\emph{Point cloud -- geometry.}
Often, the simplicial complex in the previous example is too large to work with. 
Instead one applies some geometric ideas to construct a smaller filtered simplicial complex. 
Examples include witness complexes~\cite{deSilvaCarlsson:witness}, the graph-induced complex~\cite{Dey:2013}, and the use of nudged elastic bands~\cite{adams:2011}.
These constructions typically include one or more parameters, which we append to the input to $\Comp$.
\end{example}

\begin{example}
  \label{ex:statistics}
\emph{Point cloud -- statistics.}
Instead of using geometric ideas to construct a smaller point cloud we can use statistical ideas. 
For example, one can use a kernel to smooth the point cloud to obtain a density estimator on the underlying space $Y$ and use this to filter a triangulation of $Y$~\cite{cbk:ipmi2009,bckl:nonparametric,Chen:2015b}. 
Or one may use the local density to threshold the point cloud~\cite{cidsz:mumford}; we consider this in more detail in Example~\ref{ex:DTC}.
Again, these constructions include one or more parameters, which we append to the input for $\Comp$.
\end{example}

\begin{example}
  \label{ex:regression}
\emph{Regression.}
Here we present a variant of Example~\ref{ex:HF} in which we are not given the simplicial complex. 
Instead we sample points $X = (x_1,\ldots,x_N)$, $x_i \in \R^d$ from some probability distribution on $\R^d$. We also sample corresponding perturbed function values $y_i \in \R$. For example, we may have $y_i = f(x_i) + \epsilon_i$, where $\epsilon_i$ is sampled from a univariate Gaussian.
We use $X$ to construct a Delaunay triangulation $K$.
We then use $Y = (y_1,\ldots,y_N)$ to filter $K$ as follows: $f(\sigma) = \max_{x_i\in\sigma} y_i$. This is called the lower star filtration.

Instead of the sample points lying in $\R^d$, they may lie on some compact Riemannian manifold.
See the torus example in Section~\ref{sec:intro}.

Instead of filtering $K$ directly using $Y$, one can instead use $(X,Y)$ to construct an estimator $\hat{f}$ of the unknown regression function $f$. We can then use $\hat{f}$ to filter $K$~\cite{bckl:nonparametric}.
\end{example}


\subsection{Stability}

The persistence diagram $\Dgm_p(f)$ is a summary of the function $f$, and it turns out to be a stable one. The discussion here is adapted from \cite{CohenSteiner2007}.
For a broader description, see \cite{Chazal2009b,cdsgo:book}.

For convenience, to each persistence diagram, we add every point $(r,r)$ on the major diagonal, each with infinite multiplicity.

Now suppose that $\phi: D \to D'$ is some bijection between two persistence diagrams; bijections exist because of the infinite-multiplicity points along the diagonal. The cost of $\phi$ is defined to be $C(\phi) = \sup_{u \in D} ||u - \phi(u)||_{\infty};$ that is, the largest box-norm  distance between matched points. The \emph{bottleneck distance} $W_{\infty}(D,D')$ is defined to be the minimum cost amongst all such bijections.
For example, if $D$ and $D'$ are the black and red diagrams, respectively, on the right side of Figure \ref{fig:NoisyZeroDiag}, then the best bijection would pair $u$ with $u'$, $v$ with $v'$, the two infinite-persistence points with each other, and the other two points with the closest diagonal points. The bottleneck distance is the cost of this bijection.
\begin{figure}
  \centering
  \includegraphics[scale=0.2]{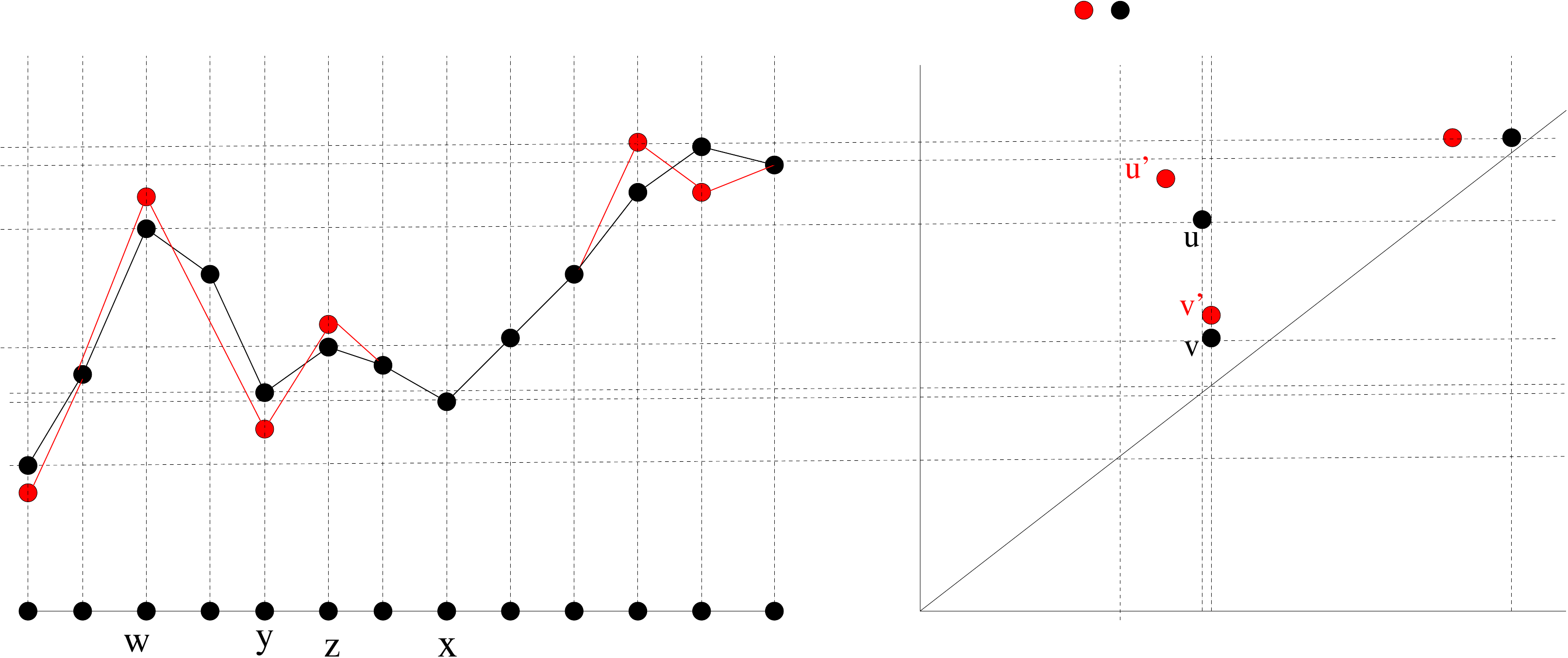}
   \caption{Left: The graphs of functions $f$ (black) and $g$ (red), both on the same domain $\mathcal{K}$.   Right: the persistence diagrams $\Dgm_0(f)$ and $\Dgm_0(g)$, using the same color scheme.
}
\label{fig:NoisyZeroDiag}
\end{figure}
The Diagram Stability Theorem~\cite{CohenSteiner2007} guarantees that persistence diagrams of nearby functions are close to one another. More precisely,
we have $W_{\infty}(D_p(f), D_p(g)) \leq ||f - g||_{|\infty}.$
This is illustrated by Figure \ref{fig:NoisyZeroDiag}.

Note that the difference between $f$ and $g$ is measured in the $L_{\infty}$ norm. In the point cloud context (Ex. \ref{ex:RF}), this translates into requiring that the two point cloud inputs be Hausdorff-close. However, the persistence diagram is not stable with respect to the addition of outliers.
We discuss this problem in more detail in Section \ref{subsec:PC} and propose a  solution in Section \ref{sec:main}.

\section{Instability}
\label{sec:IG}

The Diagram Stability Theorem tells us that the persistence diagram obtained in the output of a persistent homology computation is stable with respect to certain perturbations of the input used to construct a filtered abstract simplicial complex.
However, other outputs of persistent homology computations are not stable. This includes the simplices and cycles that generate persistent homology classes. These are of great interest to practitioners hoping to interpret persistence calculations more directly. 
In addition, many persistence computations rely on choices of parameters and the resulting persistence diagrams may be unstable with respect to these choices.

\subsection{Instability of Generating Cycles/Simplices}
\label{subsec:GCS}

Persistence diagrams are useful and robust measures of the \emph{size} of topological features. What they are less good at, on the other hand, is robustly pinpointing the \emph{location} of important topological features. 
We use Figure \ref{fig:NoisyZeroDiag} to illustrate this problem. 
Suppose that we have the fixed domain $K$ and we observe the function $f$. One of the most prominent points in $\Dgm_0(f)$ is $u$, which corresponds to the pair of values
$f(x)$ and $f(w).$ We might thus be tempted to say that $f$ has an important feature, a component of high-persistence, \emph{at} $x$.  But consider the nearby function $g$ instead. Its diagram $\Dgm_0(g)$ has a point $u'$ that is very close to $u$, but this point corresponds to the pair of values $f(y)$ and $f(w)$. 
There is still a component born at $g(x)$, but it corresponds to the much smaller persistence point $v'$.
And so while the persistence of the point $u$ is a stable summary of the function $f$, the actual location $x$ of the topological feature it corresponds to is not.

This is unfortunate. Several recent works (\cite{Bendich2015tracking}, \cite{Bendich2015trees}, among others) have shown that the presence of points in certain regions of the persistence diagram has strong correlation with covariates under study. For example, each diagram in the second cited work came from a filtration of the brain artery tree in a specific patient's brain, and it was found that the density of points in a certain middle-persistence range gave strong correlations with patient age. It would of course be tempting to hold specific locations in the brain responsible for these points with high distinguishing power.

Unsurprisingly, this problem remains for persistent homology in higher degrees. Consider Figure \ref{fig:EC} again. It is easy to see that edge $(A,D)$ creates the large loop
which corresponds to point $u \in \Dgm_1(f)$. However, a slight perturbation of the vertex configuration could render $(B,C)$ responsible for this loop instead, and so we cannot robustly locate the persistence of this loop \emph{at} $(A,D)$. 

In Section \ref{sec:main}, we both rigorously define this non-robustness and give a method for addressing it.

\subsection{Outliers and Instability of Parameter Choices}
\label{subsec:PC}

The Diagram Stability Theorem guarantees the persistence diagrams associated to two Hausdorff-close point clouds will themselves be close. However,
it says nothing about the outlier problem. For example, consider again the point cloud $X$ (Figure \ref{fig:FuzzyOutliers}, top-left) from Example \ref{ex:RF} to which we apply the Vietoris-Rips construction. 
Its persistence diagram $\Dgm_1(X)$ (top-right of same figure) has one high-persistence point, which corresponds to the ``circle'' that we qualitatively see when 
looking at the points. 
On the other hand, consider the point cloud $X'$ on the bottom-left, which consists of $X$ and three ``outlier'' points spread across the interior of the circle.
The diagram $\Dgm_1(X')$ (bottom-right) is not close to $\Dgm_1(X)$: there is still one point of fairly high persistence, but it's much closer to the diagonal than before.

In practice, this problem is often addressed by first de-noising the point cloud in some way. For example, Carlsson et. al. \cite{Carlsson2008Klein} first thresholded
by density before computing Vietoris-Rips filtrations when they discovered a Klein bottle in the space of natural images. 
There are no guarantees that a different, nearby choice of density threshold parameter would not give a qualitatively different persistence diagram.
Section \ref{sec:main} addresses this by introducing a general method for handling parameter choice in persistence computations.
\begin{figure}[ht]
\centering\includegraphics[width=.8\textwidth]{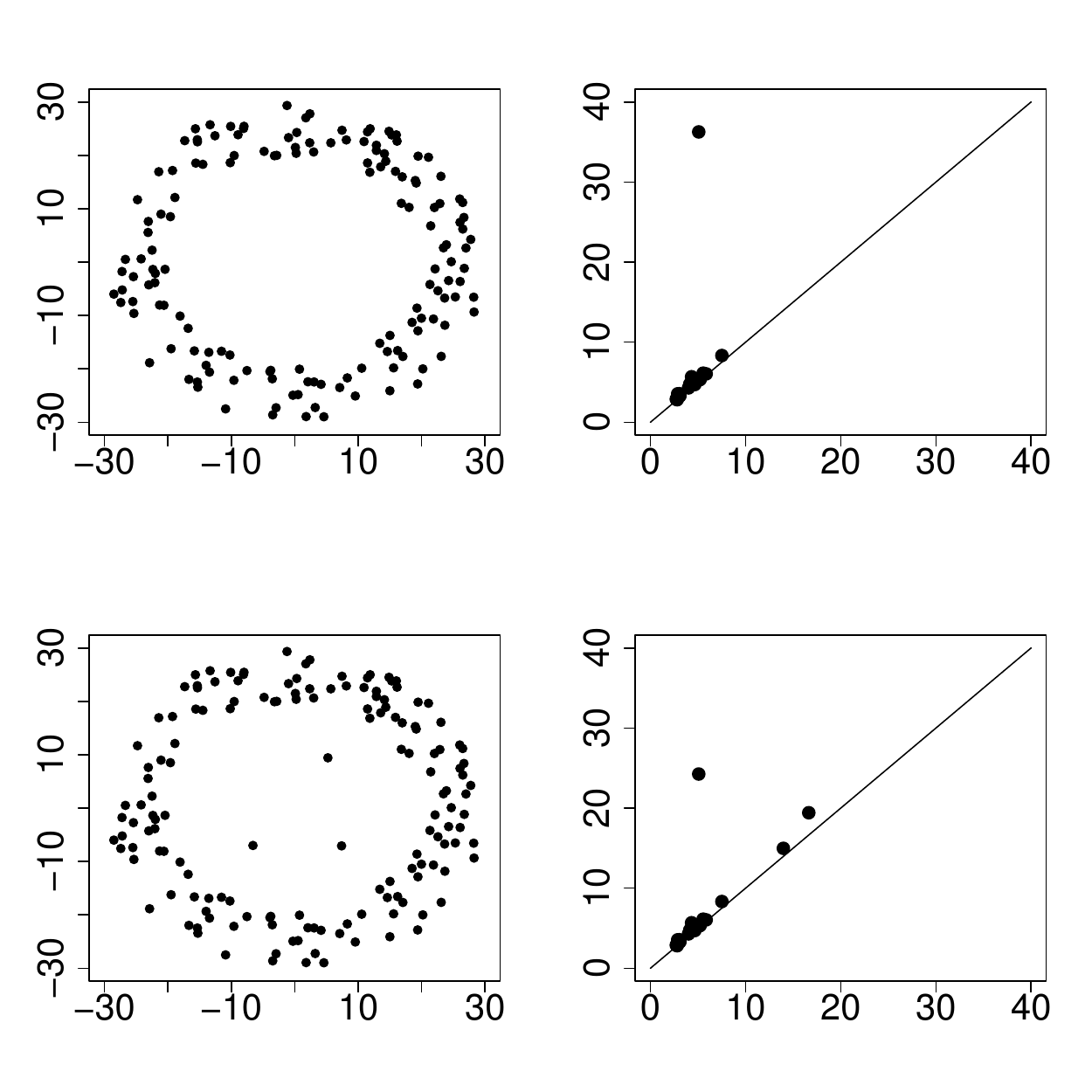}
\caption{Illustration of the outlier problem for the persistent homology of the Vietoris-Rips complex of a point cloud. Top left: $150$ points $X$, sampled from an annulus.
Top right: $\Dgm_1(X).$ Bottom left: $153$ points $X'$, which is $X$ plus three outlier points. Bottom right: $\Dgm_1(X').$
}
\label{fig:FuzzyOutliers}
\end{figure}

\section{Theory: Stability from convolutions} \label{sec:main}

In this section we show how functions may be stabilized by convolving them with a kernel. In Section~\ref{sec:ph-computations}, we will apply these results to the function $h:\R^n \to \R$ discussed in the introduction. First, we give three general results with various assumptions on the function and the kernel. Next, we apply them in three particular cases: the simple triangular kernel and the commonly used Epanechnikov and Gaussian kernels. 
We then outline a few specific examples, some of which will be explored via experiment in the next section.

\subsection{Lipschitz functions and convolution}

Let us start by recalling a few definitions.
For $C \geq 0$, a function $f:\R^n \to \R$ is said to be \emph{$C$-Lipschitz} if for all $u,v \in \R^n$, $\abs{f(u)-f(v)} \leq C \abs{u-v}$, where $\abs{x}$ denotes the Euclidean norm.
We will call a function \emph{Lipschitz} if it is $C$-Lipschitz for some $C \geq 0$.
The support of $f$, denoted $\supp(f)$, is the closure of the subset of $\R^n$ where $f$ is non-zero.

Let $h,g:\R^n \to \R$ be (Lebesgue) measurable functions that are defined almost everywhere. 
The \emph{1-norm} of $h$, is given by $\norm{h}_1 = \int_{\R^n}\abs{h(t)}dt$, if it exists.
The \emph{essential supremum} of $h$, denoted by $\norm{h}_{\infty}$, is the smallest number $a$ such that the set $\{x \st \abs{f(x)} > a\}$ has measure $0$.
If it exists,
the \emph{convolution product} of $h$ and $g$, is given by
\begin{equation*}
  (h*g)(t) = \int_{\R^n} h(s)g(t-s)ds = \int_{\R^n} h(t-s) g(s) ds.
\end{equation*}
It exists everywhere, for example, if 
one function is essentially bounded and the other is integrable;
or if 
one function is bounded and compactly supported and the other is locally integrable~\cite[Section 473D]{Fremlin2000}.

  Throughout this section we assume that $h:\R^d \to \R$ is defined almost everywhere, $K:\R^d \to \R$ and that that the convolution product $h * K$ exists almost everywhere.

\subsection{Stability theorems}
\label{sec:stability}

We now give several conditions on a pair of functions which imply that their convolution product is (locally) Lipschitz. 

The first result appears in \cite[473D(d)]{Fremlin2000}, but the proof is included here for completeness.

\begin{theorem} \label{thm:stability1}
  If $\norm{h}_1 = a$ and $K$ is $b$-Lipschitz, then $h*K$ is $ab$-Lipschitz.
\end{theorem}

\begin{proof}
  Let $g = h*K$.
  First we have, $g(u) - g(v) = \int_{\R^n} h(s) \left( K(u-s) - K(v-s) \right) ds.$ Then,
  $\abs{g(u) - g(v)} \leq \int_{\R^n} \abs{h(s)} \abs{ K(u-s) - K(v-s) } ds \leq \int_{\R^n} \abs{h(s)} b \abs{u-v} ds \leq ab \abs{u-v}.$
\end{proof}

Let $B_{\alpha}(x)$ denote the closed ball of radius $\alpha$ centered at $x \in \R^d$, and 
let $V_d$ denote the volume of the $d$-dimensional ball of radius $1$.

\begin{theorem} \label{thm:stability2}
  Let $x \in \R^d$ and let $\alpha>0$. If $\norm{h}_{\infty} \leq M$ on $B_{2\alpha}(x)$, 
$K$ is $b$-Lipschitz and
$\supp(K) \subseteq B_{\alpha}(0)$,
then $h*K$ is $2Mb\alpha^dV_d$-Lipschitz in $B_{\alpha}(x)$.
\end{theorem}

\begin{proof}
  Let $g = h*K$. Let $u,v \in B_{\alpha}(x)$.
  As in the previous proof,
  $\abs{g(u) - g(v)} \leq \int_{\R^n} \abs{h(s)} \abs{ K(u-s) - K(v-s) } ds \leq \int_{B_{\alpha}(u) \cup B_{\alpha}(v)} \abs{h(s)} b \abs{u-v} \, dx \leq 2Mb\alpha^dV_d \abs{u-v}$.
\end{proof}

\begin{theorem} \label{thm:stability3}
  If $\norm{h}_{\infty} \leq M$ and $\int \abs{K(s+t) - K(s)} \, ds \leq b\abs{t}$ for all $t \in \R^d$, then $h*K$ is $Mb$-Lipschitz.
\end{theorem}

\begin{proof}
  Let $g = h*K$. 
  Again,
  $\abs{g(u) - g(v)} \leq \int_{\R^n} \abs{h(s)} \abs{ K(u-s) - K(v-s) } ds \leq \int M \abs{K(u-v+x)-K(x)}\,dx \leq Mb\abs{u-v}$.
\end{proof}

\subsection{Application to kernels}
\label{sec:kernels}

We now apply the above theorems to smooth a function $h$, obtaining a Lipschitz function.
That is, we will take $K$ to be a \emph{kernel}, a non-negative integrable real-valued function on $\R^n$ satisfying $\int K(x) dx = 1$, $\int x K(x) dx = 0$ and $\int x^2 K(x) dx < \infty$.
For example, we can choose $K$ to be the \emph{triangular kernel}, $K(x) = c \max(1-\norm{x}, 0)$, for appropriate normalization constant $c$ (see Figure~\ref{fig:kernels}).
The most common choices are the Gaussian kernel and the Epanechnikov kernel, which are described below (see Figure~\ref{fig:kernels}).
Notice that if $K$ is a kernel, then so is $K_{\alpha}(x) = \frac{1}{{\alpha}^n} K(\frac{x}{\alpha})$.%
\footnote{More generally, 
we can choose the bandwidth to be a symmetric positive definite matrix $H$ and let $K_H(x) = \frac{1}{\sqrt{\det{H}}}K(H^{-1/2}x)$.
}
The parameter $\alpha$ is called the \emph{bandwidth} and allows one to control the amount of smoothing.

\begin{figure}
  \centering
  \includegraphics[width = \textwidth]{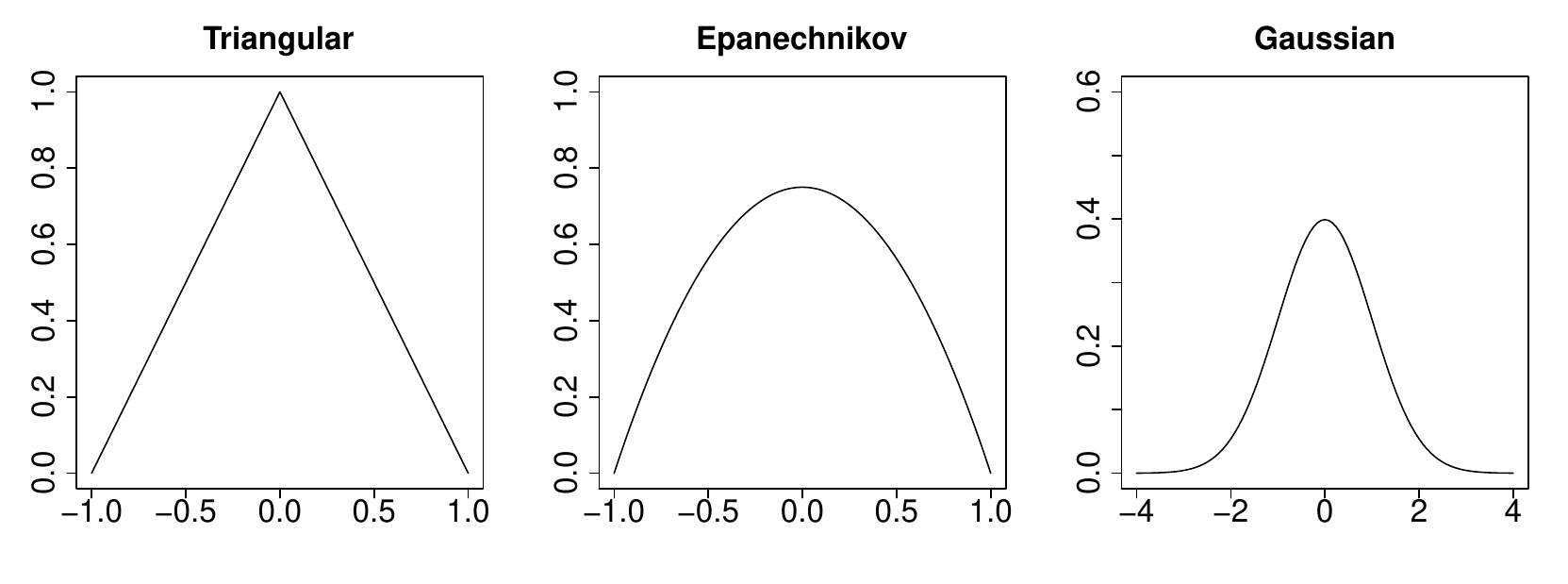}
   \caption{Graphs of three common kernels.}
\label{fig:kernels}
\end{figure}

\subsubsection*{The triangular kernel}
 
Let $\alpha>0$. Let $V_d$ denote the volume of the $n$-dimensional ball of  radius $1$. For $A \subseteq \R^d$, let $I_A$ denote the indicator function on $A$. That is, $I_A(x) = 1$ if $x \in A$ and $0$ otherwise. The \emph{triangular kernel} is given by 
\begin{equation*}
  K_{\alpha}(x) = \frac{d+1}{\alpha^d V_d} \left( 1 - \frac{\abs{x}}{\alpha} \right) I_{B_{\alpha}(0)}.
\end{equation*}
Note that $\supp(K_{\alpha})  = B_{\alpha}(0)$ and $K_{\alpha}$ is $\frac{d+1}{\alpha^{d+1} V_d}$-Lipschitz.
Applying Theorem~\ref{thm:stability2}, we have the following.

\begin{corollary} \label{cor:triangular}
  Let $x \in \R^d$. If $\norm{h}_{\infty}\leq M$ on $B_{2\alpha}(x)$ then $h*K_{\alpha}$ is $\frac{2M(d+1)}{\alpha}$-Lipschitz in $B_{\alpha}(x)$.
\end{corollary}

Note that it follows that if the bound on $h$ is global then so is the Lipschitz bound.

\subsubsection*{The Epanechnikov kernel}

Let $\alpha>0$.
The \emph{Epanechnikov kernel} is given by 
\begin{equation*}
  K_{\alpha}(x) = \frac{d+2}{2\alpha^d V_d} \left( 1 - \frac{\abs{x}^2}{\alpha^2} \right) I_{B_{\alpha}(0)}.
\end{equation*}

Now $\supp(K_{\alpha}) = B_{\alpha}(0)$ and $K_{\alpha}$ is $\frac{d+2}{\alpha^{d+1}V_d}$-Lipschitz.
Applying Theorem~\ref{thm:stability2}, we have the following.

\begin{corollary} \label{cor:epanechnikov}
  Let $x \in \R^d$. If $\norm{h}_{\infty}\leq M$ on $B_{2\alpha}(x)$ then $h*K_{\alpha}$ is $\frac{2M(d+2)}{\alpha}$-Lipschitz in $B_{\alpha}(x)$.
\end{corollary}

\subsubsection*{The Gaussian kernel}

Let $\alpha>0$. The \emph{Gaussian kernel} is given by 
\begin{equation*}
  K_{\alpha}(x) = \frac{1}{\alpha^d(2\pi)^{d/2}} e^{-\abs{x}^2/2\alpha^2}.
\end{equation*}

\begin{lemma}
  For the Gaussian kernel $K_{\alpha}$, 
  let $f(t) = \int \abs{K_{\alpha}(s+t)-K_{\alpha}(s)}\,ds$. Then $f(t) \leq \frac{2}{\alpha\sqrt{2\pi}} \abs{t}$ for all $t \in \R^d$.
\end{lemma}

\begin{proof}
  Change coordinates so that $s = -\frac{\abs{t}}{2} e_1$ and $s+t = \frac{\abs{t}}{2} e_1$. Then by symmetry
   \begin{align*}
     f(t) &= 2\left[ \int_{x_1 \geq -\frac{\abs{t}}{2}} K_{\alpha}(x) \, dx - \int_{x_1 \geq \frac{\abs{t}}{2}} K_{\alpha}(x) \, dx \right]\\
     &= 4 \int_{0 \leq x_1 \leq \frac{\abs{t}}{2}} K_{\alpha}(x) \, dx \\
      &= \frac{4}{\alpha^d(2\pi)^{d/2}} 
\int_{0}^{\frac{\abs{t}}{2}} e^{-x_1^2/2\alpha^2} \, dx_1 
\int_{-\infty}^{\infty} e^{-x_2^2/2\alpha^2} \, dx_2 \cdots
        \int_{-\infty}^{\infty} e^{-x_d^2/2\alpha^2} \, dx_d \\
          &  =\frac{4}{\alpha \sqrt{2\pi}} \int_{0}^{\frac{\abs{t}}{2}} e^{-x_1^2/2\alpha^2} \, dx_1
  \end{align*}

  It follows that $f(t) \leq \frac{4}{\alpha\sqrt{2\pi}} \int_0^{\abs{t}/2}\,dx_1 = \frac{2}{\alpha\sqrt{2\pi}} \abs{t}$.
\end{proof}

Thus by Theorem~\ref{thm:stability3} we have the following

\begin{corollary} \label{cor:gaussian}
  If $\norm{h}_{\infty} \leq M$ then $h*K_{\alpha}$ is $\frac{2M}{\alpha\sqrt{2\pi}}$-Lipschitz.
\end{corollary}

In practice, the function $h$ will rarely be essentially bounded, but this can be arranged by setting it to be $0$ outside a closed ball centered at a specified configuration.

\subsection{Sharpness of the  Lipschitz constants}
\label{sec:sharpness}

Assume that $K_{\alpha}: \R^d \to \R$ is symmetric in the first variable.
Let $h:\R^d\to\R$ be defined by $h(x) = 1$ if $x_1\geq 0$ and $-1$ otherwise.
Let $g(t) = h*K_{\alpha}(te_1) - h*K_{\alpha}(-te_1)$.
Then we calculate
\begin{align*}
  h*K_{\alpha}(te_1) &= \int_{\R^d} h(te_1-x) K_{\alpha}(x)\,dx \\
                     &= \int_{x_1 \leq t} K_{\alpha}(x)\,dx - \int_{x_1>t} K_{\alpha}(x)\,dx\\
  &= \sign(t) \int_{-\abs{t} \leq x_1 \leq \abs{t}} K_{\alpha}(x)\,dx.
\end{align*}
So\
\begin{equation*}
   g(t)= 2 \sign(t) \int_{-\abs{t} \leq x_1 \leq \abs{t}} K_{\alpha}(x)\,dx
   = 4 \sign(t) \int_{0 \leq x_1 \leq \abs{t}} K_{\alpha}(x)\,dx.
\end{equation*}

Let $K_{\alpha}$ be the Gaussian kernel. Then
\begin{align*}
g(t) &= \frac{4}{\alpha^d (2\pi)^{d/2}} \int_0^t e^{-x_1^2/2\alpha^2}\,dx_1 
\int_{-\infty}^{\infty} e^{-x_2^2/2\alpha^2} \, dx_2 \cdots
\int_{-\infty}^{\infty} e^{-x_d^2/2\alpha^2} \, dx_d \\
&= \frac{4}{\alpha\sqrt{2\pi}} \int_0^t e^{-x_1^2/2\alpha^2}\,dx_1
\end{align*}
It follows that $g(t)$
converges to $\frac{4t}{\alpha\sqrt{2\pi}}$ as $t$ approaches $0$ by the first fundamental theorem of calculus. Hence, the Lipschitz constant 
given in Corollary~\ref{cor:gaussian} is optimal. 

When $d=1$ and $K_{\alpha}$ is the triangular kernel,
$ g(t) = \frac{4\cdot 2}{\alpha V_1} \int_0^t (1-\frac{\abs{x}}{\alpha})\,dx \to \frac{4t}{\alpha}$ as $t \to 0$. So the Lipschitz constant of $h*K_{\alpha}$ is at least $\frac{2}{\alpha}$.
Hence, the Lipschitz constant given in Corollary~\ref{cor:triangular} is optimal up to at most a factor of $2$.

When $d=1$ and $K_{\alpha}$ is the Epanechnikov kernel,
$g(t) = \frac{4\cdot 3}{2 \alpha V_1} \int_0^t (1-\frac{x^2}{\alpha^2})\,dx \to \frac{3t}{\alpha}$ as $t \to 0$. So the Lipschitz constant of $h*K_{\alpha}$ is at least $\frac{3}{2\alpha}$.
Hence, the Lipschitz constant given in Corollary~\ref{cor:epanechnikov} is optimal up to at most a factor of $4$.


\subsection{Stable Computations in Practice}

Suppose that we can compute $h(x)$ for values of $x$ for which it is defined, we can sample from $K$, and that for a fixed $a \in \R^d$ 
we want to compute $g(a) = (h * K)(a) = \int_{\R^d} h(a-x)K(x) dx$. In practice, we will not be able to evaluate this integral analytically. 
We approximate $g(a)$ as follows.
Let $V$ be a random variable with probability distribution given by the kernel $K$ (one writes $V \sim K$). Let $W$ be the random variable given by $h(a-V)$. Then the expected value of $W$ is given by $E[W] = 
\int_{\R^d} h(a-x)K(x)dx = g(a)$.
We will approximate $E[W]$ by drawing a sample $\epsilon_1,\ldots,\epsilon_M$ where $\epsilon_i \sim K$ are independent. 
Then $E[W]$ can be approximated by $\overline{W}_M = \frac{1}{M} \sum_{i=1}^M h(a-\epsilon_i)$.
By the law of large numbers, $\overline{W}_M \to E[W]$, where the convergence may be taken to be in probability (the weak law) or almost surely (the strong law).
This is the justification for the computations in Section~\ref{sec:examples}.
Let us record this result.
\begin{theorem} \label{thm:lln}
Let $a \in \R^d$ and $\epsilon_1,\ldots,\epsilon_M$ be drawn independently from $K$. Then
\[
  \frac{1}{M} \sum_{i=1}^M h(a-\epsilon_i) \to g(a).
\]
\label{thm:simulation}
\end{theorem}


\subsection{Stability of the choice of kernel}
\label{sec:stability-kernel}

As should be clear, and as borne out by the experiments in Section~\ref{sec:examples}, the value of $(h*K)(\vect{a})$,
for fixed $h$ and $\vect{a}$, will certainly depend on $K$. However, there is no fragility of output with respect to this choice, as shown by the following fact. 
\begin{theorem}
 Let $h: \R^d \to \R$ be an essentially bounded function. Then the map $K \to h*K$ is Lipschitz, from $L^1(\R^d)$ to $L^{\infty}(\R^d)$.
\label{thm:choice}
\end{theorem}

\begin{proof}
  Let $\phi:L^{1}(\R^d) \to L^{\infty}(\R^d)$ be given by $\phi(K) = h*K$.
  For $x \in \R^d$,
  \[
    \abs{\left[ \phi(K) - \phi(K') \right] (x) } \leq \int \abs{h(x-t)}\,\abs{K(t)-K'(t)}\,dt \leq \norm{h}_{\infty} \norm{K-K'}_1. 
  \]
\end{proof}

\subsection{Bandwidth selection -- theoretical considerations}
\label{sec:bandwidth-theory}

After choosing a family of kernels, such as the Gaussian kernels $K_{\alpha}$ described in Section~\ref{sec:kernels}, the most important choice in implementing the method described here is the choice of bandwidth $\alpha$. 

Choosing the amount of smoothing is a well-studied problem in nonparametric regression, where increasing the bandwidth decreases the estimation variance, but increases the squared bias. Both of these terms contribute to the error. A bandwidth which optimizes this trade-off may be estimated using cross-validation. 
A proper understanding of this problem in our situation requires analysis that goes beyond the scope of the present paper. 

However, we offer some heuristics for the choice of bandwidth. First, it may be chosen to obtain a desired amount of smoothness of $h*K_{\alpha}$. For example, we may want $h*K_{\alpha}$ to be $1$-Lipschitz. Second, it seems reasonable to choose the bandwidth to (at least) equal the level of estimated noise of the input data.
One may combine these two to find the minimum bandwidth that satisfies both requirements.



\section{Application to persistent homology computations}
\label{sec:ph-computations}

Now let us apply the results of the previous section to persistent homology.
Assume we have a persistent homology computation, $\Comp$, 
with input the real numbers $a_1,\ldots,a_n$. 
If our computation is defined for all $a = (a_1,\ldots,a_n) \in \R^d$ and results in real values then we may proceed.

If not, we may reduce the more general situation to the one above as follows.
We need $h$ to be defined on all of $\R^n$ so that convolutions with a Gaussian are well defined.
Let $\Output$ be the set of outputs of this computation.
Let $D \subseteq \R^n$ be the set of all inputs for which $\Comp$ is defined.
If $D \neq \R^n$ then add a state $\emptyset$ to $\Output$ 
and say that the computation sends all points in $\R^n - D$ to $\emptyset$.
Thus we encode this computation as a function 
$H: \R^n \to \Output$. 
Let $p$ be a real-valued function on $\Output$ with $p(\emptyset)=0$.
Let $h = p \circ H: \R^{n} \to \R$.
We will need $h$ to be (Lebesgue) measurable.

To make this less abstract, we show how the instabilities described in Sections \ref{subsec:GCS} and \ref{subsec:PC} can be addressed by this method.

\begin{example}
\emph{Stable persistence located at a point.}
\label{ex:component}
We return to Example \ref{ex:HF}, where we have a geometric line graph $\mathcal{K}$ with $n$ vertices $v_1, \ldots, v_n$, and edges $e_i = (v_i, v_{i+1})$ for $i = 1, \ldots n-1.$
To produce a filtration of the type used in this example, we just need to know $n$ function values.
More precisely, our persistence computation takes as input a vector $\vect{a} = (a_1,\ldots,a_n) \in \R^n$, from which we obtain a piecewise linear function, $F_\vect{a}$, on $\mathcal{K}$ determined by
$F_\vect{a}(v_i)=a_i$.
Next we consider the corresponding abstract simplicial complex $K$ and filtration $f_\vect{a}$. 
Then we compute the persistence diagram $\Dgm_0(f_\vect{a})$.
This defines a function $H: \R^n \to \Output$, where $H(\vect{a}) = \Dgm_0(f_\vect{a})$. 

Now fix a specific vertex $x$ in $K$. A given diagram $\Dgm_0(f_\vect{a})$ either contains a point $u(x) = (b(x),d(x))$ that represents a persistent connected component born at $x$ in the filtration,
or it does not. In the former case, we define $p_x(\Dgm_0(f_\vect{a})) = d(x) - b(x),$ and in the latter we define $p(x)(\Dgm_0(f_\vect{a})) = 0$; that is, we map the diagram
to the persistence of the connected component created by the addition of this specific vertex.
Note that whether or not $p_x$ is non-zero depends on whether or not $x$ is a local minimum.

The discontinuity of the function $h_x = p_x \circ H: \R^n \to \R$ expresses the instability of localizing the persistence of a connected component. 
Referring to Figure \ref{fig:NoisyZeroDiag}, suppose that the vectors $\vect{a}$ and $\vect{e}$ produce the functions $f$ and $g$, respectively, and that the vertex $x$ is as marked in the figure. Then $h_x(\vect{a})$ is the persistence of $u$, while $h_x(\vect{e})$ is the persistence
of $v'$. 
Corollaries \ref{cor:triangular} and \ref{cor:epanechnikov}
guarantee that smoothing $h_x$ by a Triangular kernel and Epanechnikov kernel will result in a locally Lipschitz function.

To be able to convolve with the Gaussian kernel and apply Corollary~\ref{cor:gaussian} we need $h_{\sigma}$ to be essentially bounded. We can arrange this by specifying that the domain $D$ of $\Comp$ be compact and that $p_{\sigma}$ be bounded. This requires that all of the persistence pairs in the output of $\Comp$ be finite. This can be arranged by truncating at some value $M$ or by applying extended persistence \cite{cseh:extendingP}.
The resulting $h_x$ is Lipschitz. 

The experiments in Section~\ref{sec:examples} show how this works in practice.

\end{example}

\begin{example}
\emph{Stable persistence located at an edge.}
\label{ex:edge}
We return to Example \ref{ex:FDC}. In this case, $K$ is the full complex on $n$ vertices, and we start with $n$ ordered points in the plane which lead to a piecewise-linear curve $C$.
That is, $\Comp$ takes as input a vector $\vect{a} \in \R^{2n}$ and places a vertex $v_i$ at $(a_{2i-1}, a_{2i})$, thus creating a curve $C_{\vect{a}}.$ This leads to a filtration
$f_\vect{a}$ of $K$ and finally we produce $\Dgm_1(f_\vect{a}) \in \Output$. As before, $H(\vect{a}) = \Dgm_1(f_\vect{a})$ defines a function $H: \R^{2n} \to \Output.$

If we fix a specific edge $\sigma$, we can proceed as in Example \ref{ex:component} by defining the function $p_{\sigma}$ and thus
$h_{\sigma} = p_{\sigma} \circ H$. For example, taking $\sigma = (A,D)$ in Figure \ref{fig:EC} and letting $\vect{a}$ be the vector which led to that specific point configuration, we have $h_{\sigma}(\vect{a})$ equal to the persistence of $u$. As above, $g_{\sigma} = h_{\sigma}* K_{\alpha}$ is (locally) Lipschitz.

%
%

\end{example}

\begin{example}
\emph{Stable persistence of generating cycles.}
Instead of tracking which $j$-simplex creates a persistent homology class, a persistent homology algorithm may record a $j$-cycle, $\gamma$, that represents the persistence class. In this case, we can define $p_{\gamma}: \Output \to \R$ to be $d-b$ if $\gamma$ represents a persistence pair $[b,d)$ or otherwise $0$. Let $h_{\gamma} = p_{\gamma}H$ and then $g_{\gamma} = h_{\gamma} * K_{\alpha}$ is (locally) Lipschitz.
\end{example}

\begin{example}
 \emph{Stability in density-thresholding choice.}
\label{ex:DTC}
Let $Y$ be the point cloud on the bottom-left of Figure \ref{fig:FuzzyOutliers}, which we recall was created from the point cloud on the top-left by adding three outlier points. Consider any de-noising process parametrized by some real numbers.
For a specific example, let $\vect{k} = (\delta, \epsilon)$.
For each $y \in Y$, let $C_{\delta}(y) = \{x \in Y \mid ||x - y|| \leq \delta\}$.
Then define
$$
Y_{\epsilon}^{\delta} = \{y \in Y \mid \frac{|C_{\delta}(y)|}{|Y|} \geq \epsilon\}.
$$
One then applies the Vietoris-Rips construction to obtain a filtered abstract simplicial complex from $Y_{\epsilon}^{\delta}$, and then computes $\Dgm_1(Y_{\epsilon}^{\delta}).$
We may consider the input of our persistent homology computation $\Comp$ to be $a_1,\ldots,a_{2n},\delta,\epsilon$: that is, the coordinates of the vertices and the parameter values.
However, we may also take the coordinates to be fixed and only consider the parameters to be our input.
Doing this, we obtain $H: \R^2 \to \Output.$
In this case, define $p(D) = \max_{u \in D} \pers(u)$ for any degree-one diagram $D$.
Then the discontinuity of the function $h: \R^2 \to \R$ given by 
\[
\vect{k} = (\delta, \epsilon) \mapsto \Dgm_1(Y_{\epsilon}^{\delta}) \mapsto p(\Dgm_1(Y_{\epsilon}^{\delta}))
\]
expresses the instability of the threshold-parameter choice referred to in Section \ref{subsec:PC}.
If $\vect{k}$ is chosen so that the three outlier points are removed, then $h(\vect{k})$ will be the persistence of
the most prominent point on the top-right of Figure \ref{fig:FuzzyOutliers}. On the other hand, a very nearby choice of $\vect{k}$ might
fail to remove these points, and we would get the persistence of the most prominent point on the bottom-right of Figure~\ref{fig:FuzzyOutliers}.
 As above, $g_{\sigma} = h_{\sigma}* K_{\alpha}$ is (locally) Lipschitz.
\end{example}

\section{Further analysis and discussion} \label{sec:examples}

This section more deeply investigates some of the examples above
and some related practical issues.



\subsection{First line-graph experiment}
\label{sec:first-line-graph}

First we explore Example \ref{ex:component},
where the input to a persistent homology computation is a choice of function-values on the vertices
of a simplicial complex. Specifically, we consider a line graph $\mathcal{K}$ with vertices $v_1, \ldots ,v_7$, and the initial input choice
$\vect{a} = (10,11,12.5, 13, 9.9,20,1).$ The left side of Figure \ref{fig:PathGraph} shows the graph of the PL-function $F_\vect{a}$, and the persistence
diagram $H(\vect{a})$ is in the middle. Ignoring the highest-persistence point $(1,20)$, the high-persistence point $(9.9,20)$ and the medium-persistence one $(10,13)$ are created by the additions
of $v_5$ and $v_1$, respectively; that is, $h_5(\vect{a}) = 10.01$ and $h_1(\vect{a}) = 3.$ These values are of course unstable to perturbations
of $\vect{a}$: for instance, if we switch the first and fifth entries of $\vect{a}$, the reader can check that $h_5((9.9,11,12.5,13,10,20,1)) = 3.$ 
\begin{figure}[ht]
  \centering \includegraphics[width = \textwidth]{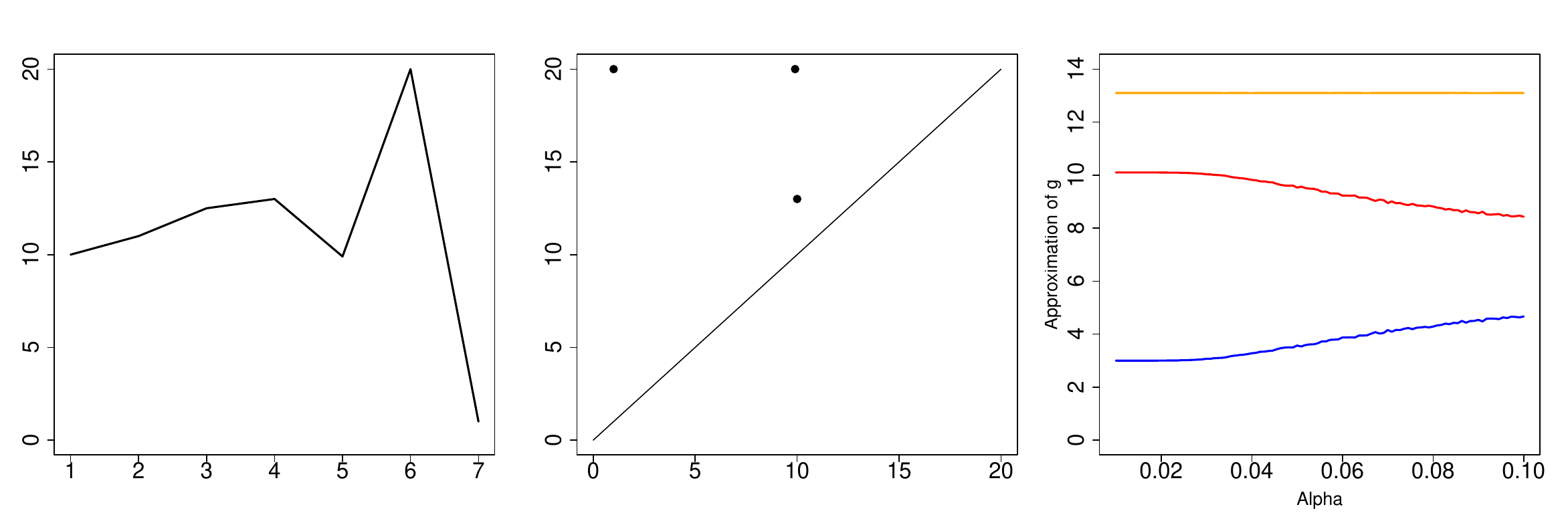}
\caption{Input and results of first line graph experiment Left: a graph of the function $F_\vect{a}$ defined on a line graph with seven vertices. Middle: the persistence diagram
$H(\vect{a}) = \Dgm_0(f_\vect{a})$. We follow the extended persistence convention and pair the global min with the global max.
Right: results of the experiment. Middle graph shows the value of $g_{5,\alpha}(\vect{a})$ versus $\alpha$; bottom graph shows $g_{1,\alpha}(\vect{a})$ versus $\alpha$; top graph shows their sum.}
\label{fig:PathGraph}
\end{figure}
Let $K_{\alpha}$ be a seven-dimensional Gaussian kernel with mean at the origin and bandwidth $\alpha$. For each $i = 1, \ldots, 7,$ put
$g_{i,\alpha} = h_i * K_{\alpha}$. 
The right side of Figure \ref{fig:PathGraph} shows graphs of the approximate values of $g_{5,\alpha}(\vect{a})$ and $g_{1,\alpha}(\vect{a})$, plotted against $\alpha$, as well as a graph of their sum. There were $100$ evenly spaced values of $\alpha$ used, ranging from  $\alpha = 0.01$ to $\alpha = 0.1$.
To make these graphs, we followed the approximation procedure based on Theorem \ref{thm:simulation}. For each fixed $\alpha$, we took $N = 10000$ independent
draws $\vect{\epsilon}_1, \ldots \vect{\epsilon}_{10000}$ from $K_{\alpha}$, and computed
$$
g_5(\vect{a}) \approx \frac{1}{10000} \sum_{i=1}^{10000} h_5(\vect{a} + \vect{\epsilon}_i),
$$
with an identical procedure for $g_1(\vect{a}).$ Figure \ref{fig:PathGraph} also plots $g_1(a) + g_5(a)$ to demonstrate that, for the bandwidths considered, vertices $v_1$ and $v_5$ are together responsible for a consistent total amount of persistence but are competing for which generates the high persistence point.

\subsection{Second line-graph experiment}

Again we explore Example \ref{ex:component}, this time with the input $\vect{a} = (5,1.1,1,1.05,15)$ to a persistent homology computation that builds a filtration
on a line graph with five vertices $v_1, \ldots v_5$. The function $F_\vect{a}$, whose graph is on the left of Figure \ref{fig:FlatGraph}, has a global min at $v_3$.
From the diagram in the middle, we see $h_3(\vect{a}) = 15 -1 = 14$. Note that $h_i(\vect{a}) = 0$ for $i \neq 3$, since only one component
is created during the entire filtration. On the right, we see convolved values of these functions for $100$ evenly spaced choices of bandwidth between $0.01$ and $0.5$, with notation and computation procedure exactly as in \ref{sec:first-line-graph} above.
\begin{figure}[ht]
\centering \includegraphics[width = \textwidth]{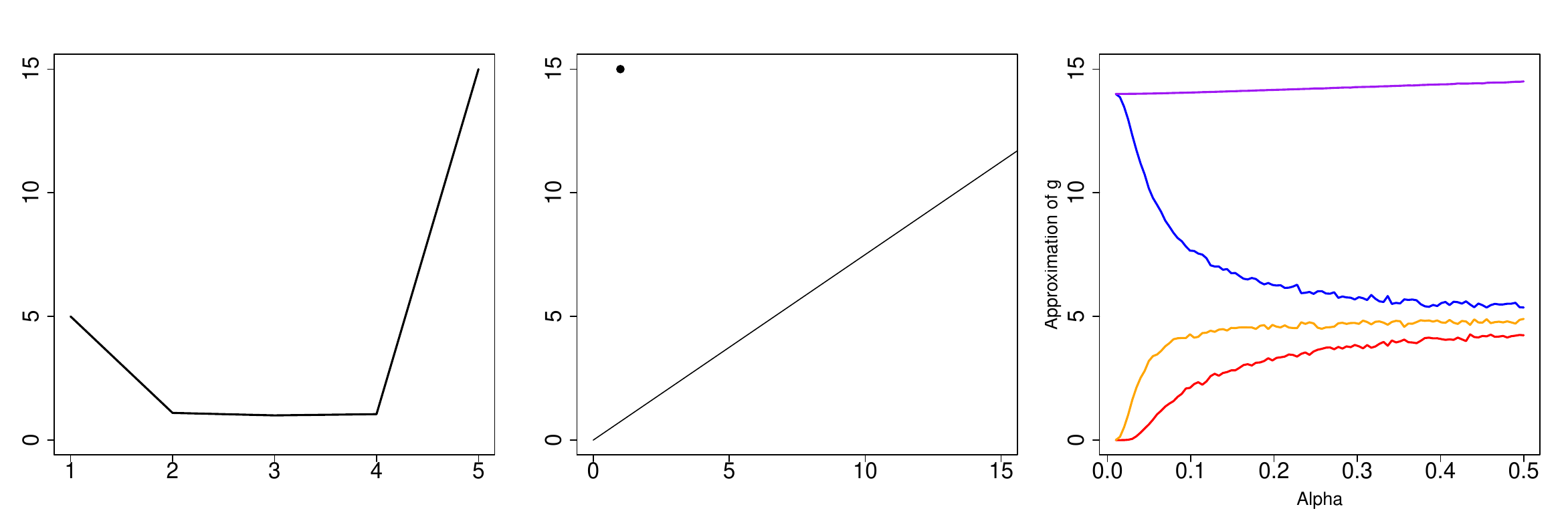}
\caption{Input and results of second line graph experiment Left: a graph of the function $F_\vect{a}$ defined on a line graph with seven vertices. Middle: the persistence diagram
$H(\vect{a}) = \Dgm_0(f_\vect{a})$. We follow the extended persistence convention and pair the global min with the global max.
Right: results of experiment. Moving from bottom to top, the values of $g_{2,\alpha}(\vect{a}), g_{4,\alpha}(\vect{a}), g_{3,\alpha}(\vect{a})$, and their sum, all
plotted against $\alpha.$ }
\label{fig:FlatGraph}
\end{figure}

\subsection{Distance-to-a-curve experiment}
\label{sec:curve}

Next we reconsider Example \ref{ex:edge}.
Let $C$ be the PL-curve with nine vertices on the left side of Figure \ref{fig:SmallBigDiamond}.
In our language, $C = C_\vect{a}$, where the input vector $\vect{a}$ specifies the coordinates of the nine vertices:
$v_1 = (0,0.1), v_2 = (1,1), v_3 = (2,0.12), v_4 = (7,5), v_5 = (12,0), 
v_6 = (7,-5), v_7 = (2,-0.12), v_8 = (1,-1), v_9 = (0,-0.1).$
Following the vocabulary of Example \ref{ex:FDC}, this curve placement leads to a order-preserving function $f_{\vect{a}}$ on the abstract
full complex $K$ on nine vertices. 

Its degree-one persistence diagram $H(\vect{a}) = \Dgm_1(f_{\vect{a}})$, in the middle of the same figure, has only two off-diagonal points.
The first, at $(0.2,10)$, is created by the positive edge between $v_1$ and $v_9$, while the second, at $(0.23,2)$, comes
from the edge between $v_3$ and $v_7$.
Thus we have $h_{1,9}(\vect{a}) = 9.8$ and $h_{3,7}(\vect{a}) = 1.77$. As usual, these values are highly unstable to small perturbations in the vertex positions.
 \begin{figure}[ht]
\centering \includegraphics[width = \textwidth]{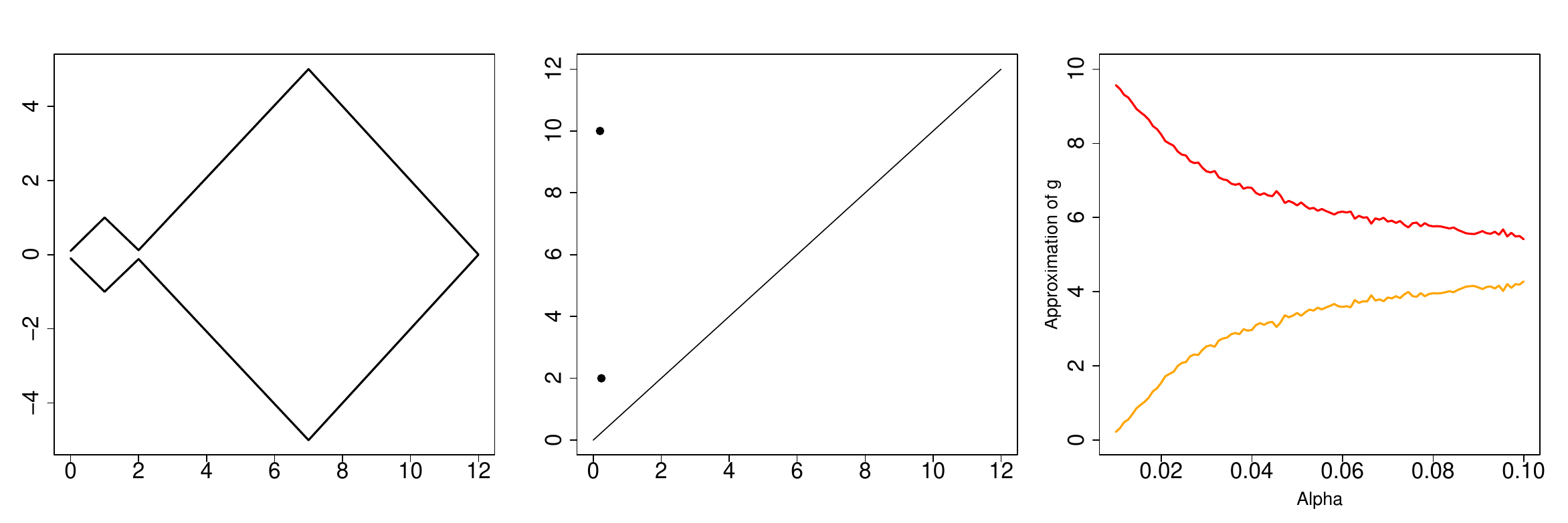}
\caption{Input and results of distance-to-curve experiment. Left: the PL plane curve $C_{\vect{a}}$ whose nine vertices are defined in the text. Middle: the persistence diagram
$H(\vect{a}) = \Dgm_1(f_{\vect{a}})$. 
Right: Results of experiment. Top graph shows the value of $g_{1,9,\alpha}(\vect{a})$ versus $\alpha$, bottom graph shows $g_{3,7,\alpha}(\vect{a})$ versus $\alpha$.}
\label{fig:SmallBigDiamond}
\end{figure}
In very similar fashion to the last experiment, we then computed approximate values at $\vect{a}$ for the convolved functions 
$g_{1,9,\alpha} = h_{1,9} * K_{\alpha}$ and $g_{3,7,\alpha} = h_{3,7} * K_{\alpha}$
where $K_{\alpha}$ was an eighteen-dimensional Gaussian kernel with bandwidth $\alpha$. The results appear on the right side of Figure \ref{fig:SmallBigDiamond}.



\subsection{Locating a point in the domain}

Let $u = (9.9,20)$ be one of the high-persistence points in the diagram for our first line-graph experiment. It is accurate to say that $u$ was created, for this specific persistent homology computation, by the addition of $v_5$ to the filtration. However, it is also a potentially misleading thing to say.

We propose that the difference between the persistence of $u$ and the values of the convolutions $g_{5,\alpha}(\vect{a})$ might be seen as an indicator for how confidently one should locate $u$ at $v_5$. The graphs on the right side of Figure \ref{fig:PathGraph} tell us that this confidence should be low. On the other hand, the other high-persistence point $w = (1,20)$ is created by the addition of $v_7$.
It turns out that $g_{7,\alpha}(\vect{a})$ remains very close to $19$ for all $\alpha$ within a reasonable range.

\subsection{Spreading out a point in the domain}

Alternatively, one might choose to give $u$ a more fuzzy location. A reasonable idea would be to spread out its location between vertices $v_5$ and $v_1$, since $v_1$ is responsible for creating the same component in a nearby filtration. The graphs in Figure \ref{fig:PathGraph} bear this out: note that the sum of the two convolution values $g_{5,\alpha}(\vect{a}) + g_{1,\alpha}(\vect{a})$ is always very close to the sum of the persistences of the components created by $v_1$ and $v_5$.
Similarly, in the second line-graph experiment, it would be reasonable to smear the location of the only point throughout the immediate neighborhood of $v_3$.

\subsection{Convolved values as features}

One could also use the values of $g_i$ or $g_{i,j}$ as features in a machine-learning scheme. That is, the vector $(g_{1,\alpha}(\vect{a}), \ldots, g_{7,\alpha}(\vect{a}))$ could be used as a summary feature of both the filtration created by $\vect{a}$ and the noise model $K_{\alpha}$. The stabilities offered by Corollary \ref{cor:gaussian} and Theorem \ref{thm:choice} make this an appealing option.

\subsection{Bandwidth selection -- in practice}
\label{sec:bandwidth-practice}

Our procedures depend on a free parameter, the bandwidth. For example, 
we chose a bandwidth of $3$
in Example~\ref{ex:double-annulus}.
Here we consider a slightly simpler example and consider the effect of varying the bandwidth.

We sample 1000 points uniformly from an annulus of inner and outer radius $20$ and $40$. 
Using Dionysus~\cite{dionysus}, we compute the $1$-dimensional persistent homology of the alpha complex of our sample and obtain a representative cycle for the longest bar. However, the embedded location of this cycle is unstable. We would like to quantify and visualize the uncertainty of this location. To do so, we consider a grid of squares with edge-length $1$. We perturb the sampled points 1000 times by adding Gaussian noise and find the proportion of trials in which the representative cycle produced by Dionysus intersects each square. By performing this procedure simultaneously for every square in the grid, we obtain Figure \ref{fig:bandwidth}. When the standard deviation of the Gaussian noise is very small, the representative cycle barely changes between perturbations, resulting in a small number of squares with high probability of intersecting the cycle. As the standard deviation increases, the picture becomes more diffuse.

\begin{figure}
  \centering
  \includegraphics[width=\textwidth]{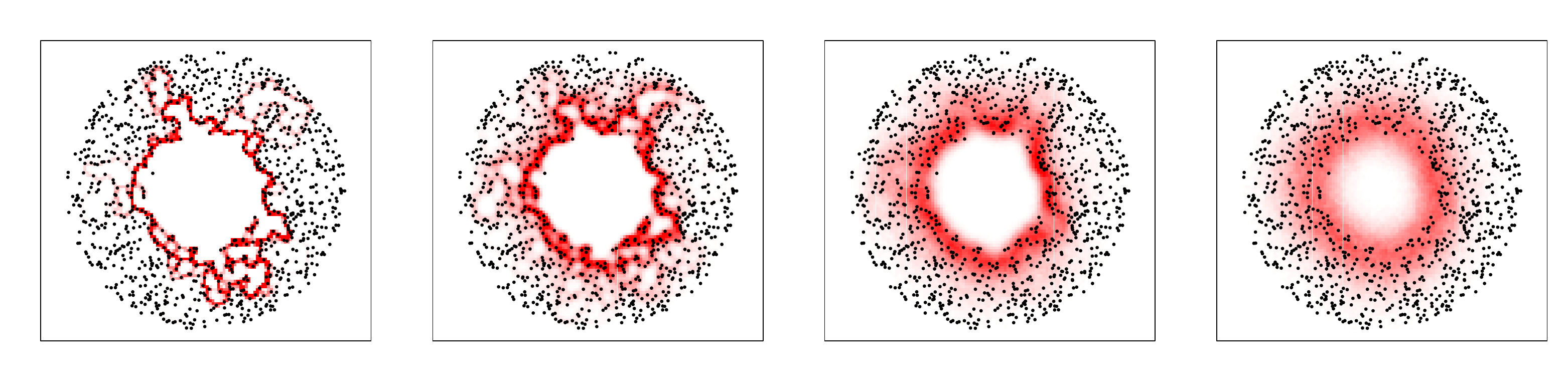}
  \caption{The probability of intersecting a representative cycle for each square in a grid. The bandwidths of the Gaussian kernel are $0.2$, $1$, $3$, and $10$. The color scale is given in Figure \ref{fig:doubleannulus}.}
  \label{fig:bandwidth}
\end{figure}

\subsection{Possible choices for the location of the generator in the brain imaging data}
\label{sec:brain-location}

In this section and the following section we expand on Example~\ref{ex:brain-artery}, which applies our method to real data.

We can obtain the estimate of $h*K$ for the observed data in Example~\ref{ex:brain-artery} for a ball of any radius by considering the distance from the location of the generator of the $28$th longest bar in one of the iterations of Algorithm~\ref{alg:main} to the location of the generator of the $28$th longest bar in the observed data,
and computing the empirical cumulative distribution function of this function. See Figure~\ref{fig:ecdf}.
In this figure we see that there is a competitor for the location of the generator of the $28$th longest bar.
So in fact, the situation with this real data is quite similar to that in the elementary synthetic examples in Sections \ref{sec:first-line-graph} and \ref{sec:curve}. 

\begin{figure}
  \centering
  \includegraphics[width=0.4\textwidth]{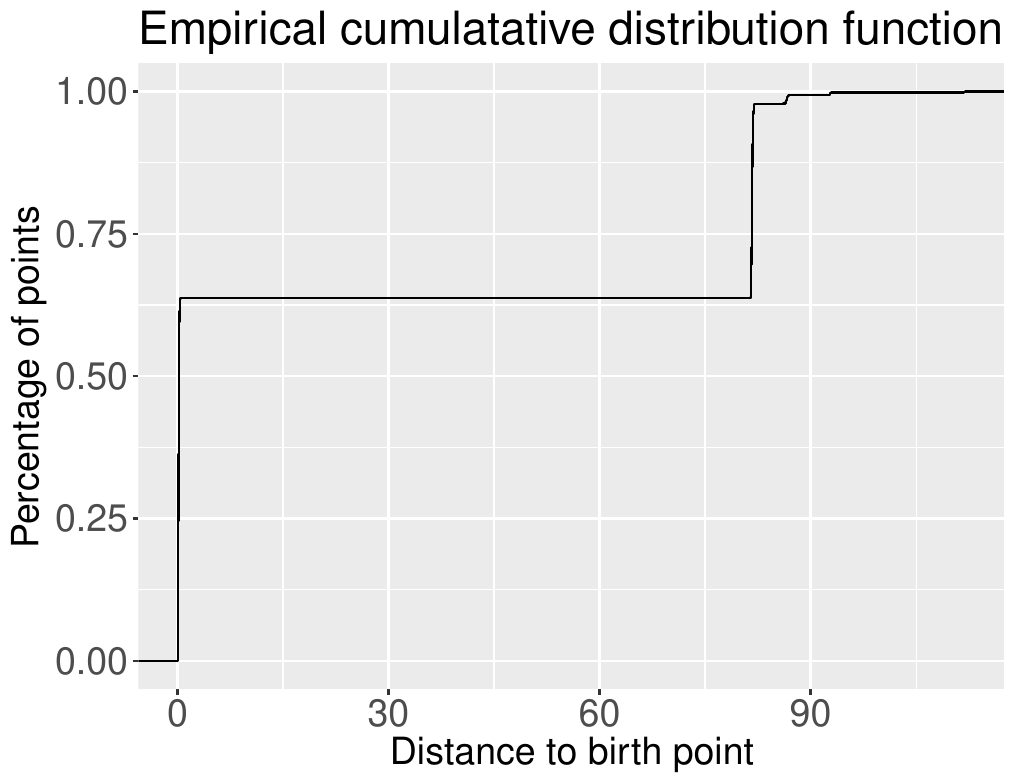}
  \caption{The empirical cumulative distribution function of the distance between the the location of the generator of the $28$th longest bar in one of the iterations of Algorithm~\ref{alg:main} to the location of the generator of the $28$th longest bar in the observed data.}
  \label{fig:ecdf}
\end{figure}

\subsection{A concrete use of our brain imaging analysis for the practitioner}
\label{sec:concrete-use} 


In Example~\ref{ex:brain-artery}, we show that under certain small perturbations, nearly two-thirds of the generators for the 28th longest bar in degree-zero persistent homology are born in a particular ball. Using the results of Section~\ref{sec:brain-location} we see that there is another ball containing a substantial proportion of the likely generators for this persistent homology class. 

By repeating this computation for all subjects, we allow the clinician wanting to compare subjects with respect to some clinically important variable (in this case age) to focus their attention on a couple of small regions of each image. In particular, the clinician can compare the morphology of the brain arteries in these neighborhoods.

\section{Future work}
\label{sec:Disc}

The work presented here opens many interesting questions and possible directions for future research. 

\subsubsection*{Machine learning}
Persistence diagrams have been used to produce features for machine-learning and statistical methods. This paper takes a first step towards the extraction of stable features that describe much of the other information produced during a persistent homology computation.
For example, one could apply the ideas presented here to construct topological features for machine learning that are not only based on critical values but also on the locations of critical points.

\subsubsection*{Visualization}
The ideas of this paper could be used to build visualization tools. For example, one might want to compute a persistence diagram, click on a point, and have the possible location candidates shown on the domain, perhaps with some sort of heat map of likelihood.

\subsubsection*{Convergence results}
What is the rate of convergence of Algorithm~\ref{alg:main} as the number of repetitions $M$ increases?

\subsubsection*{Samples with increasing numbers of points}

In all of the examples we consider, the number of sampled points is fixed. Expand the framework presented here to allow for increasing sample sizes, allowing asymptotic results to be considered.

\subsubsection*{A continuous theory}
Finally, we also hope to enrich the theory whose development has started here. 
Work in the category of metric spaces, to define
some versions of the functions $h_x$ and $g_x$, and to prove Lipschitz-continuity of the latter. We believe that
Example \ref{ex:regression} points us in the right direction.

\subsubsection*{Features in other vector spaces}
The present paper has only considered the stabilization of real-valued functions. However, one could consider functions in $\R^n$ or more generally into Banach spaces.

\subsection*{Acknowledgments}
The authors would like to thank Justin Curry, Francis Motta, Chris Tralie, and Ulrich Bauer for helpful conversations and to thank the referees for suggestions that improved the paper. The first author would like to thank the University of Florida for hosting him during the initial phase of this research.
The second author would like to acknowledge that this research was partially supported by the Southeast Center for Mathematics and Biology, an NSF-Simons Research Center for Mathematics of Complex Biological Systems, under National Science Foundation Grant No. DMS-1764406 and Simons Foundation Grant No. 594594, and that this material is based upon work supported by, or in part by, the Army Research Laboratory and the Army Research Office under contract/grant number W911NF-18-1-0307.


\end{document}